\documentclass{article}
\usepackage{ijcai24}

\usepackage{times}
\usepackage{soul}
\usepackage{url}
\usepackage[hidelinks]{hyperref}
\usepackage[utf8]{inputenc}
\usepackage[small]{caption}
\usepackage{graphicx}
\usepackage{amsmath}
\usepackage{amsthm}
\usepackage{booktabs}
\usepackage[switch]{lineno}

\usepackage[ruled]{algorithm2e} % For algorithms
\usepackage{amsfonts}
\usepackage{multirow}
\usepackage[dvipsnames]{xcolor}

\newtheorem{theorem}{Theorem}
\newtheorem{corollary}{Corollary}
\newtheorem{proposition}{Proposition}
\newtheorem{lemma}{Lemma}
\newtheorem{definition}{Definition}

\makeatletter
\newcommand\footnoteref[1]{\protected@xdef\@thefnmark{\ref{#1}}\@footnotemark}
\makeatother

\title{Vulnerabilities of Single-Round Incentive Compatibility in Auto-bidding: Theory and Evidence from ROI-Constrained Online Advertising Markets}

\author{
	Juncheng Li
	\And
	Pingzhong Tang
	\affiliations
	Tsinghua University\\
	\emails
	\{lijuncheng13, kenshinping\}@gmail.com
	}

\newcommand{\abs}[1]{\left| #1 \right|}
\newcommand{\bigparen}[1]{\left( #1 \right)}
\newcommand{\bigbraces}[1]{\left\{ #1 \right\}}
\newcommand{\bigbrackets}[1]{\left[ #1 \right]}

\newcommand{\argmax}{\operatornamewithlimits{arg\,max}}

\begin{document}

\maketitle

\begin{abstract}
    Most of the work in the auction design literature assumes that bidders behave rationally based on the information available for every \textit{individual}  auction, and the revelation principle enables designers to restrict their efforts to \textit{incentive compatible} (IC) mechanisms.
    However, in today's online advertising markets, one of the most important real-life applications of auction design, the data and computational power required to bid optimally are only available to the platform, and an advertiser can only participate by setting performance objectives and constraints for its proxy \textit{auto-bidder} provided by the platform.
    The prevalence of auto-bidding necessitates a review of auction theory.
    In this paper, we examine the markets through the lens of ROI-constrained value-maximizing campaigns.
%    which are widely adopted in many global-scale online advertising platforms.
%    With theoretical analysis and empirical experiments on both synthetic and realistic data
    We show that second price auction exhibits many undesirable properties (computational hardness, non-monotonicity, instability of bidders' utilities, and interference in A/B testing) and loses its dominant theoretical advantages in single-item scenarios.
%    Some of these phenomena have been identified in literature (for budget-constrained auto-bidders) and widely observed in practice, and we show that they are actually deeply rooted in the property of (single-round) incentive compatibility.
	In addition, we make it clear how IC and its runner-up-winner interdependence contribute to each property.
%    Although many complex designs have been proposed in literature, first and second price auctions remain popular in industry.
    We hope that our work could bring new perspectives to the community and benefit practitioners to attain a better grasp of real-world markets.
\end{abstract}

\section{Introduction}

Auto-bidding has become a corner stone of modern advertising markets.
For better end-to-end performance and customer experience, platforms now provide algorithmic agents to set fine-grained bids for advertisers, who only need to submit campaign-level optimization objectives and constraints.
As a result, the community has seen a surge of publications on auto-bidding in recent years.

One of the most notable features of the auto-bidding paradigm is the change of roles played by advertisers and platforms.
Traditionally advertisers are assumed to bid rationally for each individual ad slot, and thus the real-time bidding (RTB) literature focuses on developing algorithms for advertisers to maximize their objectives subject to different constraints and auction rules, while the auction design literature, anticipating the best response of advertisers or RTB algorithms, explores new auction rules to optimize various goals of the platform.
However, in auto-bidding markets, most of the technical components
are under the management of the platform:
auto-bidders provided by the platform will compete with each other, based on valuations predicted by the platform, under auction rules that are also designed by the platform.
%The prevalence of auto-bidders necessitates a review of auction theory from a broader point of view.
Such a greater control over the market imposes a more diverse set of requirements upon the designer.
The auction design literature has long been focusing on incentive compatibility and welfare/revenue guarantee, but the desiderata of an auto-bidding mechanism are far beyond these two.
Google AdSense's \textit{partial} shift at 2021 from second to first price auction\footnote{
	In November 17, 2021, Google moves the AdSense auction for Content, Video, and Games from second price auction to first price, while keeping Search and Shopping as before \cite{google2021shift}.}
might serve as an exemplary demonstration of this perplexity.

Despite the evolution of ecosystems, first and second price auction remain as the dominant mechanisms used in practice.
In this paper, we study the mathematical model abstracted from real-world first and second price auction markets with ROI-constrained auto-bidders.
In addition to theoretic interests, our work are also motivated by real-world observations.
Traditional interpretations of some phenomena may confuse both sides of the market and lead to business choices detrimental in the long run.
Our goal is thus to develop a deeper understanding of auto-bidding at the market scale, and provide practitioners a more holistic view to facilitate decision making.
We will show, through a series of theoretical and empirical results, that the dominant  advantage of second price auction over first price is, in many regards, reversed in the world of auto-bidding.
%Some of our results extend the counterparts for budget-constrained auto-bidders \cite{chen2021throttling,chen2021complexity,conitzer2021pacing,conitzer2021multiplicative}, and this, along with many more results that are less discussed in literature, leads to the crux of our work: the vulnerabilities deeply rooted in the desideratum of utmost importance in single-item auction design: incentive compatibility (IC).
It will be made clear throughout the journey how IC, instead of simplifying the reasoning of both bidders and auctioneers, unnecessarily complicates the game in a profound way.

\subsection{ROI-Constrained Auto-bidding  Markets and Related Works}

%Deviated from the classic quasi-linear model, the objective of auto-bidders is to acquire as much value (e.g., measured in the number of conversions) as possible subject to various constraints.\footnote{Budget-constrained auto-bidders are typically utility-maximizers. However, this can be interpreted as a marginal ROI-constraint for each individual auction.}
Previously much effort was spent on the study of auto-bidders with budget constraints \cite{karande2013optimizing,charles2013budget,balseiro2019learning,gao2020first,chen2021complexity,balseiro2021budget,conitzer2021pacing,conitzer2021multiplicative}, but Return-on-Investment (ROI), another widely adopted auto-bidding option, received much less attention until very recently \cite{golrezaei2021auction,balseiro2021landscape,deng2021towards}.
For campaigns with ROI-constraints, advertisers should submit either a target ROI, a target Return-on-Ad-Spend or a target Cost-per-Action (tROI/tROAS/tCPA).\footnote{\label{footnote:terminology}See Appendix \ref{app:terminologies} for more notes on terminology.}
%(tROI or tROAS, measured by the ratio of total quasi-linear utility or value to total spend, respectively).
The objective of the auto-bidder is then to maximize the acquired value while keeping the average spend for each unit of value below the target threshold.
ROI-constrained auto-bidders are dominating in market share in many regions of the world. % e.g., the Chinese display advertising market.
There is also recent empirical evidence \cite{golrezaei2021auction} from Google AdX that advertisers are indeed ROI-constrained.

An advertiser's value for an ad slot (to which we will refer as a generic \textit{good} throughout the rest of the paper) is typically given by the product of the conversion-rate (predicted by the platform) and the value of each conversion.
For ROI-constrained campaigns, the latter part is the amount of money that the advertiser is willing to pay the platform for each (unit of) conversion.\footnoteref{footnote:terminology}
Truthful bidding for each individual auction is ex-post optimal for non-constrained quasi-linear utility-maximizing bidders.
But for ROI-constrained value-maximizers, it is possible to raise bids above values to win more while keeping the average spend of each conversion below the threshold.
In this paper, auto-bidders will take values as given and be restricted to the \textit{multiplicative pacing} strategy,\footnoteref{footnote:terminology} wherein the bids of each bidder could only be generated through scaling the values for all goods by a common multiplier of its choice.
There is always a multiplier that is ex-post bidder-optimal in second price auction, and the strategy is one of the most implemented in the industry regardless of auction formats (in particular, it remains popular in first price auction\footnote{Possibly surprisingly, we will show in Section \ref{sec:first_vs_second} that it will \textit{not} bring incentive issues for first price auction and it is in the interests of a platform to enforce so. In contrast, both bidders and sellers have incentives to deviate from it in second price auction (Appendix \ref{app:arbitrary_bid}).}).

Table \ref{tab:one} positions our market model within the growing literature on auto-bidding.
Existing works differ in the modeling of advertisers' valuations and utilities.
Traditionally the valuation of each bidder is modeled as being drawn from a stochastic process independently of each other.
%\cite{balseiro2019learning,balseiro2021budget,balseiro2021landscape,golrezaei2021auction}.
We adopt a framework first studied by Conitzer et al. \shortcite{conitzer2021multiplicative,conitzer2021pacing}, where the valuation is given as the (fixed) input and can be viewed as a discretized density function or a realized sample of an arbitrary joint distribution.
In particular, the correlation prevalent in advertising can thus be captured.
For utilities, previous research mainly focuses on budget-constraints,
%\cite{balseiro2019learning, gao2020first, balseiro2021budget, conitzer2021pacing, conitzer2021multiplicative, chen2021complexity}.
and our paper complements the recent line of work on ROI-constraints.
%\cite{deng2021towards, balseiro2021robust, golrezaei2021auction}.
Babaioff et al. \shortcite{babaioff2021non} also consider ROI, but they model bidding behaviors with respect to the marginal ROI of each individual auction rather than the average ROI across.
\begin{table*}%
	\centering
	\begin{tabular}{p{3.8cm}p{6cm}p{6cm}}
		\toprule
		\multirow{2}{*}{Utility model}
		& \multicolumn{2}{l}{Valuations}
		\\ \cmidrule(l){2-3}
		& Deterministic, correlated
		& Stochastic, independent  \\
		\midrule
		ROI-constrained $\qquad$ value-maximizer
		& Our model, Balseiro et al. \shortcite{balseiro2021robust} & Balseiro et al. \shortcite{balseiro2021landscape}, Golrezaei et al. \shortcite{golrezaei2021auction} \\
		Budget-constrained utility-maximizer
		& Conitzer et al. \shortcite{conitzer2021pacing,conitzer2021multiplicative}, Chen et al. \shortcite{chen2021complexity}, Gao and Kroer \shortcite{gao2020first} & Balseiro and Gur \shortcite{balseiro2019learning}, Balseiro et al.  \shortcite{balseiro2021budget} \\
		ROI\&budget-constrained value-maximizer
		& Deng et al. \shortcite{deng2021towards}, Aggarwal et al.  \shortcite{aggarwal2019autobidding} \\
		\bottomrule
	\end{tabular}
	\caption{Common auto-bidder types and valuation assumptions in the literature. \label{tab:one}}
\end{table*}%

Deng et al. \shortcite{deng2021towards} and its follow-up work \cite{balseiro2021robust} study a model that largely coincide with ours.
Their goal is to design mechanisms having revenue and welfare guarantees when the designer has fairly accurate signals on the valuation.
Though we will also report some results on revenue and welfare, we emphasize that, for today's large scale auto-bidding systems,  auction (combined with auto-bidding strategies) acts more like an efficient distributed algorithm to match demand with supply and compute market-clearing prices.
With this mentality, our work covers a broader range of properties and focuses heavily on their possible practical impacts on advertisers and platforms.
We also differ in the adopted solution concept.
We allow fractional allocation and incorporate the tie-breaking rule into the solution concept, which is not only well-motivated, but also guarantees the existence of equilibrium even though the market is discrete and discontinuous.
In contrast, Balseiro et al. \shortcite{balseiro2021robust} break ties  lexicographically and more complex auctions like VCG and GSP are considered there. So they choose a weaker solution concept called undominated bids and avoid the discussion of existence.
Nonetheless, this distinction diminishes for large markets as the result of a single auction becomes negligible.
%We actually rely on this to compute approximate equilibria on realistic data.

Auto-bidding or RTB algorithms have been studied for a long time.
Such works typically assume a \textit{stationary} environment and optimize various objectives for a \textit{single} advertiser.
They fail to capture other bidders' responses invoked by the action of the focal agent, and the resulting equilibrium outcome may not fulfill the initial design goal if all the bidders implement the same strategy (see, e.g., Appendix \ref{subsec:ab_testing_results}).
One notable exception is the work by Aggarwal et al. \shortcite{aggarwal2019autobidding}, who were aware of this problem and tried to prove the existence of an equilibrium.
But their treatment of equilibrium is incomplete (see a discussion in Appendix \ref{app:more_related_works}).

Researchers have shown that the computation of equilibrium bears some inherent hardness with budget-constrained multiplicative pacing bidders.
For computing any equilibrium, we improve previous result to show that it is PPAD-hard to approximate within \textit{constant} parameters (for budget-constraints, it was shown to be hard to approximate within \textit{polynomially small} parameters \cite{chen2021complexity}).
Our result is built on a more concise reduction quite distinct from the previous one.
We also improve the NP-hardness result of optimizing revenue/welfare \cite{conitzer2021multiplicative} to APX-hardness.
Moreover, the source of these hardness and the vulnerabilities of IC are demonstrated more clearly with our constructions, which we believe could help both researchers and practitioners better extrapolate our techniques and insights.

\subsection{Single-round Incentive Compatibility}

% in general, revelation principle
The classic revelation principle ensures us that, when designing single-item auctions, any implementable allocation rule could be implemented in an incentive compatible way by directly eliciting bidders' private information (in most cases only the valuations to the item to be sold are needed).
By focusing on IC mechanisms, the designer loses nothing while bidders could be prevented from strategic behaviors.
In comparison, the characterization and computation of equilibrium in first price auction is notoriously hard (see, e.g., \cite{filos2021complexity} and the survey therein).

% auto-bidding, ex-post buyer optimal
Though second price auction is only IC for single-item auctions, it possesses another fascinating property in auto-bidding markets: from a single bidder's perspective, each individual auction comes with a \textit{winning price} independent of its own bid (which is essentially an equivalent statement of IC).
As a consequence, the (ex-post/offline) task of each auto-bidder is a linear program for bidders with linear constraints like ROI and budget.
The optimal solution can be well approximated by the multiplicative pacing strategy, which is simple to implement and performs well even in the online setting \cite{balseiro2019learning,balseiro2022best}.
This provides second price auction an illusory strategyproofness that could be called \textit{ex-post IC}: every advertiser could truthfully report its tROI and happily accept the equilibrium bids given by its proxy auto-bidder as deviating unilaterally will not bring extra profit \textit{at the moment}.

% externality
Nonetheless, even long before the advent of the auto-bidding era, experiments have revealed that bidders' behaviors in second price auction are far from truthful (see, e.g., \cite{kagel2011auctions}).
Recall that, in second price auction, the price is \textit{set} by the runner-up, but \textit{paid} by the winner.
In some senses, both the winner and the runner-up care \textit{little} about the absolute magnitudes of their \textit{own} bids: only the ranking (first and second) is important.
This can also be seen from the linear program optimizing utility for the auto-bidder (see, e.g., \cite{aggarwal2019autobidding}), where no decision variable denoting bids appears and the solution only prescribes whether each auction should be won or not.
However, the bid of the runner-up always means \textit{a lot} to the winner.
A well-known enemy of IC due to this \textit{runner-up-winner interdependence} is externality, i.e., the utility of a bidder depends on the allocation and payment of not only itself, but also others.
%\footnote{See Appendix \ref{app:more_related_works} for more on externality.}

It should not be surprising that externality makes manipulation worthwhile since IC is not designed for the job.
In our auto-bidding markets, the objective of each auto-bidder is defined clearly without externality.\footnote{Another common scenario where IC fails is that bidders do not have complete knowledge of their valuations. This is not an issue either in our case.}
However, auto-bidders adjust bids based on the overall performance across \textit{all} auctions.
Even though at equilibrium changing bids unilaterally could not benefit the manipulator immediately, it may trigger a cascade of responses that shift the whole market state.
This creates a kind of externality that is  \textit{internalized} into the outcome of the shifted equilibrium.
It is worthwhile to point out that, for some results in this paper, it seems to be the multiplicative pacing strategy that leads to a property.
Actually it is irrelevant of the specific bidding strategy or even the ROI-constraint: for simultaneous auctions with single-round IC, bidders will always bid strategically across all auctions,\footnote{Even a non-constrained utility-maximizer could manipulate if one of its opponents has constraints. Truthful reporting is secure only when all the other bidders are insensitive to each other's bid.} which is enough to establish those results.

Letting the (bids of) opponents determine the payment of the winner is the key to establish IC for single-item auctions, but it is also the key to open the Pandora's box in auto-bidding markets, as will be detailed in the remainder of this paper.

\subsection{Contributions}
In this paper, we consider the auto-bidding market where ROI-constrained value-maximizing bidders compete with each other in simultaneous auctions.
Our main focus is second price auction: behaviors of auto-bidders within first price auction markets will be discussed at the very end.
%We will give a comprehensive comparison between them at the very end.

We start by formulating our own solution concept of the market (and its approximate version), named \textit{auto-bidding equilibrium}, since a pure Nash equilibrium may not always exist.
The equilibrium reasonably captures the expected steady-state that the auto-bidders intend to reach collectively, and its guaranteed existence puts our study on a solid theoretical footing.
% the most important parts are two complexity results and non-monotonicity example
% apart from their own importance, the proof and construction reveal many insights that demonstrate the vulnerability of IC, help readers understand and interpret the rest of the paper, and extrapolate outside
Our most important results are as follows:
\begin{itemize}
	\item It is PPAD-hard to find an approximate auto-bidding equilibrium within constant parameters.
	\item It is APX-hard to optimize revenue or welfare over all auto-bidding equilibria.
	\item Non-monotonicity: an advertiser who raises its tROI/tROAS (or equivalently, lowers the tCPA) at the equilibrium could end up with a higher revenue after deviation (through a natural equilibrium transition process to resolve multiplicity).
\end{itemize}
Besides their own significance, the constructions used in establishing these results reveal many crucial structures of the market and highlight the role played by IC.
They serve as the foundation to comprehend and interpret other characteristics of the market.

%attempt to answer questions pertinent to the interests of not only researchers, but also practitioners.
We proceed to explore several practical concerns of great consequence.
For advertisers, we show that the market suffers severe utility instability and input sensitivity issues.
For platforms, we demonstrate that biases exist broadly in the widely-used A/B testing.
Finally, we give a comprehensive comparison between first and second price auction, and show that first price auction is generally exempted from the above undesirable properties and performs better.
As an application of our results, we give our guess about why Google AdSense moves from second to first price auction in a partial manner.

%We start in Section \ref{sec:auto_bidding_equilibrium} by introducing our mathematical model of auto-bidding markets and equilibria.
%Section \ref{sec:vulnerabilities} examines the vulnerabilities of IC from four perspectives, namely computational complexities, exploitability (by both advertisers and platforms), utility instability for advertisers, and interference in A/B testing for platforms.
%In particular, we try our best to pinpoint how the IC property contributes to every theoretical or empirical phenomenon such that researchers and practitioners could better extrapolate from the reasoning and intuition we provide to more complex markets in both theory and the real world.
%We conclude in Section \ref{sec:first_vs_second} with a comprehensive comparison of first and second price auction both within and beyond auto-bidding.
%As an application of the results and intuitions developed in Section \ref{sec:vulnerabilities}, we also give our guess about why Google AdSense moves from second to first price auction for only Content, Video and Games, but not Search and Shopping.
%We hope that our work could bring new perspectives to the community, and inspire practitioners to pay closer attention to the IC property and attain a better grasp of real-world markets.

\section{Markets and Equilibria}
\label{sec:auto_bidding_equilibrium}

We consider a market where a set of bidders $N = \{1, \dots, n\}$ compete for a set of divisible goods $M = \{1, \dots, m\}$.
Without loss of generality, the tROIs of all bidders are set to zero (equivalently, tROAS of one),\footnote{See Appendix \ref{app:equilibrium_definition_rationale} on why this is WLOG. Later we may sometimes \textit{change} the tCPA of a bidder by a factor $\lambda$, by which we mean all the valuations of this bidder is scaled by $\lambda$.} i.e., each bidder's spend should be no more than its acquired value.
We use $v_{i, j}$ to denote the value of bidder $i$ to good $j$.
For each good $j$, there is at least one bidder $i$ such that $v_{i, j} > 0$.
The platform simultaneously runs a single-item second price auction for every good.
Auto-bidders are restricted to apply multiplicative pacing strategies: the action space of bidder $i$ is the set of undominated multipliers $\alpha_i \in [1, A]$,\footnote{Fixing other bidders' bids, $\alpha_i = 1$ always weakly dominates any $\alpha_i < 1$. The cap $A$ can safely be replaced by $+\infty$ for most cases. See more discussion in Appendix \ref{app:equilibrium_definition_rationale}.} and its bid for good $j$ is $\alpha_i v_{i, j}$.
Multiplicative pacing is ex-post bidder-optimal, as shown in Proposition \ref{proposition:expost_buyer_optimal}.
The omitted proof follows from a linear programming formulation of the ROI-constrained value maximization problem,
and the result is widely known in both the literature and the industry.
\begin{proposition}\label{proposition:expost_buyer_optimal}
	Suppose that bidders can bid arbitrarily across auctions.
	Holding all other bidders' bids, each bidder has a best response wherein bids are generated by scaling its valuations of all goods by a uniform multiplier, given that it could freely choose to win any fraction of a good of which it is a tied winner.
	%	\footnote{In real-world systems, fractional allocations can really happen as a result of small fluctuations of pacing multipliers of several closely competing auto-bidders.}
\end{proposition}

To complete the picture, one's first intuition may be to specify a tie-breaking rule, make the allocation $x \in [0, 1]^{n \times m}$ (where $x_{i, j}$ is the fraction of good $j$ allocated to bidder $i$) uniquely determined by $\alpha$, define bidder's utility as the ROI-constrained valuation:
\begin{displaymath}
	u_i(\alpha) =
	\left\{
	\begin{array}{ll}
		\sum_j  x_{i, j} v_{i, j}, & \text{if } \sum_{j} x_{i, j} p_j \leq \sum_{j} x_{i, j} v_{i, j};
		\\
		-\infty, & \text{otherwise;}
	\end{array}
	\right.
\end{displaymath}
and study the pure Nash equilibrium (PNE) of the game.
However, in Appendix \ref{app:equilibrium_definition_rationale}, we will give an example where, no matter how ties are broken, a PNE does not exist.
On closer inspection, we find that the steady-state of the market can be captured instead by the solution concept defined below.

%Our foremost task is to characterize a reasonable steady-state of the market and justify its existence such that the discussion could be put on a solid footing.
%Since  a pure Nash equilibrium (PNE) may not always exist,\footnote{This is in contrast to the budget-constrained case \cite{conitzer2021multiplicative} where the equilibrium is a \textit{refinement} of PNEs (with a flexible tie-breaking rule that favors an equilibrium if one exits). This justifies that ROI-constraint deserves a separate treatment from the very beginning.} we have to define our own solution concept: auto-bidding equilibrium (and the approximate version), at which, fixing all the others' multipliers, an auto-bidder either dominates all the auctions it participates or its ROI-constraint is binding.

\begin{definition} \label{def:autobidding_equilibrium}
	An \textbf{auto-bidding equilibrium} $(\alpha, x)$ consists of multipliers $\alpha \in [1, A]^n$ and allocation $x \in [0, 1]^{n \times m}$ such that
	\begin{itemize}
		\item bidders with the highest bid win the good:
		if $x_{i, j} > 0$, $\alpha_i v_{i, j} = \max_k \alpha_k v_{k, j}$ for all $i, j$;
		\item winner pays the second price:
		if $x_{i, j} > 0$, then $p_j = \max_{k \neq i} \alpha_k v_{k, j}$ for all $i, j$;
		\item full allocation of goods:
		$\sum_i x_{i, j} = 1$ for all $j$;
		\item ROI-feasible:
		$\sum_{j} x_{i, j} p_j \leq \sum_{j} x_{i, j} v_{i, j}$ for all $i$;
		\item maximal pacing:
		%		(1, non-binding ads)
		unless $\alpha_i = A$, $\sum_{j} x_{i, j} p_j = \sum_{j} x_{i, j} v_{i, j}$ for all $i$.
		%		(2, non-winning ads)
		%		if $x_{i, j} = 0, \forall j$, then $\alpha_i v_i r_{i, j} = \max_k \alpha_k v_k r_{k, j}$ for some $j$;
		%		\item maximal pacing: fix $\alpha_{-i}$, $\alpha_i$ maximize $u_i$ defined above.
	\end{itemize}
\end{definition}

The auction rules and ROI-constraints are directly imposed by the first four conditions, while the best responses among bidders are encoded in the maximal pacing condition in a less straightforward way.
To see this, note that, from bidder $i$'s perspective, given the winning price of each good, $\alpha_i$ acts as a marginal-ROI threshold:
it will win all goods in whole with a marginal ROI strictly larger than $\frac{1}{\alpha_i} - 1$ and lose those with ROI strictly lower.
At any auto-bidding equilibrium, if $\alpha_i > 1$, $\alpha_i$ (and the resulting vector of bids with $b_i = \alpha_i v_{i, j}$) is exactly a best response since the marginal ROI of any good lost or tied is strictly lower than zero, and winning anymore will definitely violate the ROI-constraint.
If, however, $\alpha_i = 1$ and bidder $i$ is only allocated a \textit{fraction} of some good as in the example given in Appendix \ref{app:equilibrium_definition_rationale}, $\alpha_i$ is technically not a best response (for the normal form game) but the maximal pacing condition is still satisfied.
For such bidders, their opponents could always oscillate their multipliers around the equilibrium point to achieve the corresponding stable allocation.

\begin{theorem} \label{thm:existence}
	An auto-bidding equilibrium always exists.
\end{theorem}

The proof is in Appendix \ref{app:proof_existence}.
In addition, the definition and existence result can be extended to incorporate \textit{reserve prices} and \textit{additive boosts} (see Appendix \ref{app:reserve_and_boosts}), both of which are common practice in the industry.

The relaxed version of the equilibrium, named \textbf{$(\eta, \delta)$-approximate} auto-bidding equilibrium, is defined as follows:
\begin{itemize}
	\item bidders with bids close enough to the highest can win the good: 
	if $x_{i, j} > 0$, $\alpha_i v_{i, j} \geq (1 - \eta) \max_{k} \alpha_k v_{k, j}$ (if there is a reserve price $r_j$, winner's bid should also be no less than $(1 - \eta) r_j$);
	\item winner pays the second price (even if it is the bidder with the second highest bid; if there is a reserve price, the price is the maximum of the reserve price and the second price);
	\item full allocation of goods;
	\item approximately ROI-feasible:
	$\sum_{j} x_{i, j} p_j \leq (1 + \delta) \sum_{j} x_{i, j} v_{i, j}$	for all $i$;
	\item approximately maximal pacing:
	%		(1, non-binding ads)
	unless $\alpha_i = A$, $\sum_{j} x_{i, j} p_j \geq (1 - \delta) \sum_{j} x_{i, j} v_{i, j}$ for all $i$.
\end{itemize}

\section{Computational Complexities}

In this section, we demonstrate two different kinds of intractability or unpredictability of the market.
%You may wonder why we bother with the complexity of finding the optimal equilibrium if it is already hard to find any. Actually they demonstrate different kinds of instability or unpredictability.
Our first result shows that, in general, it is hard for a market to reach a stable state.
But even for cases where an equilibrium is easy to achieve, the difference among equilibria may be large, and the second result tells us that, in general, it is hard to determine how large the difference is.

%The two computational tasks share the same source of hardness.
At a high level, the interdependence that the (bid of) runner-up determines the payment of the winner endows the market with a structure similar to Boolean operators (electronic components and conductive wires). In digital electronics, a functionally complete set of operators can be assembled to compute any Boolean function.
In our model, it is PPAD-hard to find an equilibrium since it can encode any stable state of a circuit that is continuous. When optimizing revenue or welfare, the circuit structure collapses to discrete choices that correspond to equilibrium selection, and the problem becomes NP-hard (and APX-hard since the correspondence is almost exact).
%Computational hardness is established through reduction, which deals with encoding an arbitrary instance of a problem (in our case, finding certain states of a circuit) using a properly constructed instance of the problem which is to be shown at least as hard as the former.
Note that hardness does not only exist in the family of instances constructed in the reduction: for a generally hard problem, it is possible but usually non-trivial to identify a meaningful subset of instances that are computationally tractable.

Our focus here is the many complex structures brought into the market by IC.
%The design of practical algorithms is not the focus of this paper.
Nonetheless, to explore equilibrium properties quantitatively, we develop two algorithms (see Appendix \ref{app:algorithms} for details) to compute equilibrium.
Their own properties are beyond the scope of this paper.
%A noteworthy observation is that the iterative method works fairly well for realistic datasets consisting of more than ten million auctions. This seems to indicate that an equilibrium is easy to achieve for realistic value distribution, but we will argue in Appendix EC.8 that it might not be the case and this performance should result more from the multi-staged auction mechanism used by the platform from which our data are taken.

\subsection{Complexity of Finding Any Equilibrium}
\label{subsec:ppad}

\begin{theorem} \label{thm:ppad_hardness}
	It is PPAD-hard to find an $(\eta, \delta)$-approximate auto-bidding equilibrium for some constant $\eta, \delta > 0$.
\end{theorem}

The full proof is in Appendix \ref{app:proof_ppad_hardness}.
We will use equilibria in a properly constructed market to encode feasible states of a circuit, which is one of the most fundamental and frequently used objects in complexity theory \cite{chen2009settling,rubinstein2018inapproximability}.
Papadimitriou and Peng \shortcite{papadimitriou2021public} show that a circuit consisting only of a continuous version of NAND (NOT-AND) gate is enough to capture PPAD-hardness,  in analogy to the \textit{functional completeness} of NAND gate in digital electronics.
The continuous gate computes the function that sums all (at most 3) inputs and inverts the result:
for a gate $u$, given all the input values $y_w$ of its incoming gates $w \in N_u$, its own value $y_u$ should satisfy that
\begin{displaymath}
	y_u \in \left\{
	\begin{array}{ll}
		\{0\}, & \text{if } \sum_{w \in N_u} y_w > 0.5; \\
		\{1\}, & \text{if } \sum_{w \in N_u} y_w < 0.5; \\
		{[0, 1]}, & \text{otherwise}.
	\end{array}
	\right.
\end{displaymath}
An assignment $y$ of values to gates is feasible if the above constraints are satisfied for all gates.
The key construction of our reduction is given in Table \ref{tab:ppad_main_construction_exact}. (This is a simplified version: in the full proof we will remove the reliance on reserve prices and take approximation into account.)
\begin{table}[t]
	\centering
	\begin{tabular}{c ccc c c}
		\toprule
		goods & $w_1$ & $w_2$  & $w_3$  & $\underline{u}$ & $\bar{u}$
		\\
		\midrule
		input bidder $w_1$ & 1/14
		\\
		input bidder $w_2$ & & 1/14
		\\
		input bidder $w_3$ & & & 1/14
		\\
		gate bidder $u$ & 1/3 & 1/3 & 1/3 & 1/2 & 1/4
		\\
		reserve price & $1/7$  & $1/7$ & $1/7$ & 1 &  1
		\\
		\bottomrule
	\end{tabular}
	\caption{Valuation profile to encode gate $u$.
		%		Goods shown in the table have no other interested bidders outside. The gate bidder $u$ will also act as input bidder for other gates. An input bidder is simply another gate bidder who will be interested in their own gate goods and input goods of other gates.
	}
	\label{tab:ppad_main_construction_exact}
\end{table}
We will associate each gate $u$ with a bidder of the same name and use its multiplier $\alpha_u$ to encode the gate's value $y_u$ (with $y_u = \alpha_u / 2 - 1$).
The valuation profiles are calibrated such that bidder $u$ would always win all input goods $w_1, w_2, w_3$ but at varying prices, determined by (multipliers of) their corresponding  input bidders.
The correspondence between bidder $u$ and gate $u$ is almost straightforward:
\begin{itemize}
	\item If the sum of input prices is too high, to satisfy its ROI-constraint, bidder $u$ could not even win good $\underline{u}$ in whole and thus won't raise $\alpha_u$ above $2$.
	\item If the sum is too low, good $\bar{u}$ is required to satisfy bidder $u$'s appetite for more valuation, which forces $\alpha_u$ to be fixed at 4.
	\item If input prices sum to exactly 0.5, bidder $u$ can freely choose $\alpha_u$ as long as it is allocated $\underline{u}$ in whole but not $\bar{u}$ at all.
\end{itemize}

Intuitively, the hardness of finding a feasible state of a circuit lies in the coordination among the interconnected gates.
Given any value assignment, for each unsatisfied gate $u$, the corresponding $\alpha_u$ in the constructed market is either too large (input goods are too expensive and ROI-constraint is violated) or too small (input goods are too cheap and bidder $u$ has the incentive to win $\bar{u}$).
Consider hypothetically that we apply a naive search algorithm\footnote{Theoretically, any problem in PPAD can be reduced to the generic End-Of-The-Line problem (by which the class is defined), where we are given (1) a directed graph consisting solely of non-intersecting directed paths (lines) and (2) a vertex with no predecessor (the start of a line). The task is to find a vertex with no successor (the end of a line). There is a natural algorithm (inevitably inefficient if PPAD $\neq$ P) that simply searches along any path. The search procedure we depict here shares a similar spirit, but it may (or may not) circulate and is only used to give some intuition.} where, for each non-equilibrium assignment $y$, we choose some unsatisfied $u$, lower $\alpha_u$ by a small amount if it is too large, and raise it if it is too small.
As we adjust $\alpha_u$ (or $y_u$) for a bidder (gate), the payments (sums of inputs) of its outgoing neighbors change accordingly, which may also change their directions of adjustment (e.g., a bidder goes from ROI-feasible to infeasible).
%And we may follow the path to adjust multipliers for these newly infeasible bidders, their outgoing neighbors may also become infeasible.
As gates can be assembled arbitrarily, we can imagine how hard it is to find a state that satisfies the constraints of all bidders/gates.

In retrospect, IC requires the payment of the winner to be determined \textit{externally} by other bidders.
As a result, nothing is local in the market and we can connect bidders in a way that encodes any circuit perfectly.
You may think that the market constructed in the reduction can be simplified, e.g., by setting a reserve price for the input good such that its price would no longer be affected by the input bidder and an equilibrium could be easier to find.
However, this requires much knowledge of the specific market \textit{a priori} such as an upper bound of a bidder's multiplier, which is typically impossible.
In general markets, the aforementioned chasing behavior is even more complex: e.g., if a gate bidder lowers its multiplier, the consequence is simply lowering payments for its outgoing neighbors and increase their ROIs, but in general markets it may also lose goods it previously won and instead decrease ROI for the opponent who now wins the item with a negative marginal ROI.
See the non-monotonicity instance in Section \ref{sec:non_monotonicity} for an example of this complex cascading phenomenon.

% where does the hardness come from
% coordination among gates/bidders
% bidder does the gate job: low prices (low multiplier of others) -> high bid, which is only possible with IC in place since the price is determined externally
% but it is not completely free
% nothing is local
% in general, non-winning relationships are hard to identify as easily as in our handcratfted instances, see non-monotonicity for an example

\subsection{Complexity of Finding Revenue or Welfare Optimal Equilibrium}

\begin{theorem} \label{thm:complexity}
	It is APX-hard to find the optimal revenue or welfare over all auto-bidding equilibria.
\end{theorem}

The full proof is in Appendix \ref{app:proof_apx_hardness}.
When dealing with optimization problems, besides encoding a \textit{feasible} circuit state into equilibrium, we also need to relate the \textit{objective} (revenue/welfare) of the market owner to one that is hard to optimize within the circuit.
Interestingly, it is still the encoding of feasibility that reveals more distinctive structures of the problem.\footnote{Here feasibility should be encoded in a way such that all the equilibria are known and easily enumerable. The hardness comes from finding the best (not any) one.}
The latter step only involves simple operations (max and sum) that are naturally embedded in the auction rules (winner pays the \textit{largest} non-winning bid) and bidder's rationale (\textit{aggregating} outcomes across all auctions).

The problem to be reduced is of discrete nature, and we will encode a 0-1 \textit{choice} using equilibrium \textit{multiplicity} as shown in Table \ref{tab:choice}.
\begin{table}[t]
	\centering
	\begin{tabular}{c ccp{3cm}}
		\toprule
		valuation &  good $1$ & good $2$  & goods that both bidders will never win
		\\
		\midrule
		bidder $1$ & $2$ & $1$ & \multicolumn{1}{c}{$\epsilon$}
		\\
		bidder $2$ & $1$ & $2$ &  \multicolumn{1}{c}{$\epsilon$}
		\\
		\bottomrule
	\end{tabular}
	\caption{A symmetric market with extremely asymmetric equilibria.}
	\label{tab:choice}
\end{table}
Here the market is symmetric in valuation but has equilibria that represents two extremes in allocation: bidder $i$ (1 or 2) wins both good 1 and 2 with $\alpha_i \geq 2$ and $\alpha_{-i} = 1$.
There is another equilibrium where $\alpha_1 = \alpha_2 = 2$ and each bidder gets the good it values the most, which is natural, fair and revenue-optimal within the sub-market consisting of good 1 and 2.
But we will see in the full proof that it is often advantageous for the seller to choose the asymmetric equilibrium as it frequently dominates the symmetric one in revenue \textit{of the whole market}.
The reason lies in the last column of Table \ref{tab:choice}: though bidder 1 and 2 will never win those goods, they are price-setters and it is usually profitable to enforce a high multiplier on one, rather than letting them share the sub-market fairly but both bid at a moderate level.

The reduction clearly demonstrates the ``externality'' created by IC: each pair of bidders in the sub-market characterized by Table \ref{tab:choice} determines the clearing prices of some other auctions they will never win (though such knowledge is hard to acquire a priori in practice).
On the other hand, if the seller is able to engage in the choice of equilibrium, it may favor those unfair outcomes that, though less profitable locally, could drive up revenue from goods outside these sub-markets.
See Appendix \ref{subsec:seller_competition} and \ref{app:ab_testing} for further discussion on externality among sub-markets and Appendix \ref{app:arbitrary_bid} on how sellers could prevent those unwanted equilibria by actively elevating the bid landscape without breaking single-round IC.

\section{Non-monotonicity}
\label{sec:non_monotonicity}

With auto-bidding built into the mechanism, advertisers are effectively playing a meta-game through reporting tROIs.
In single-item first or second price auctions, raising one’s bid will never decrease its winning probability. In real-world markets, advertisers also expect ROI monotonicity, i.e., lowering tROI/tROAS (raising tCPA) should bring them more valuation.
Due to equilibrium multiplicity, however, it is not clear how to define utility functions  for  the advertiser game, let alone monotonicity.
We avoid this technicality by examining a proper equilibrium transition process, from which we can see how the runner-up-winner interdependence triggers the chain reaction that is  complex and counter-intuitive.

The deviation of the manipulator and the transition of equilibrium work as follows.
At round 0, the manipulator $i_0$ changes its tCPA to a fraction $r < 1$ of the original, resulting in a valuation profile $v'$ such that $v'_{i_0, j} = r v_{i_0, j}, \forall j$.
Each bidder then applies an iterative method (see Appendix  \ref{app:iterative}) to optimize their utilities and collectively find the new equilibrium.
The behavior of the algorithm is very intuitive: if the current ROI (aggregated over a moving window of recent rounds) is too high, lower the multiplier, and vice versa.
To make the transition smooth, $\alpha_{i_0}$ is divided by $r$ right after the tCPA modification such that the fine-grained bids of the manipulator are kept unchanged at the moment.

The detail of the non-monotone example is in Appendix \ref{app:non_monotonicity}.
Here we give a high-level description of the process.
At the old equilibrium, good 1 is the only good won by $i_0$, and it is shared between $i_0$ and another bidder $i_1$.
$i_0$ initiates the dynamics by lowering $\alpha_{i_0}$, since its ROI-constraint is now violated after the update of tCPA.
However, $i_1$ does not want to win good $1$ completely, otherwise its ROI-constraint will also be violated. So $i_1$ lowers its multiplier as well, which further triggers the same behavior for bidder  $i_2$. As a result, $i_0$, $i_1$ and $i_2$ reach an almost perfect coordination where the multiplicative ratios among their multipliers remain nearly constant all the way through the transition.
There is another bidder, $i_3$, who pays less due to the lowered second prices set by $i_0$, $i_1$ and $i_2$.
Therefore it tries to win more goods by gradually raising its multiplier.
During the process, $i_2$ and $i_1$ pay more for goods whose second prices are set by $i_3$, and thus they have to give up goods of which they are one of the tied winners (these goods have the lowest marginal ROIs): $i_2$ gives up good $2$ to $i_1$, and $i_1$ wins more good $2$ but loses good $1$ to $i_0$ to balance its deficit, which contributes to the success of the manipulation of $i_0$. In the end, $i_3$ takes a fraction of good $2$ away from $i_1$ to bind its ROI-constraint, and $i_1$ compensates this by taking a fraction of good $1$ from $i_0$. Nonetheless, $i_0$ still benefits from lowering its tCPA.

\section{Practical Properties of the Market}

\subsection{Utility Instability for Advertisers}
Besides high-quality value estimation and bid optimization, platforms are also trying to serve many other needs of their clients, among which  utility (i.e., the total acquired value) stability stands out because (1) advertisers expect a smooth experience, and more importantly (2) utility is the most prominent feedback on how successful their advertising campaigns are.
As a result, utility instability may bring confusions and put many good campaigns at the risk of being forfeited prematurely.
Our experiments (see Appendix \ref{app:multiplicity}) show that, in markets generated from several different stochastic processes, a large utility gap between the worst and the best equilibrium for an advertiser is quite often to be observed, and it is fairly common that an advertiser wins nothing in some equilibrium but acquires a significant positive value in others.
The gap seems to reduce for thicker markets, but large-scale realistic instances suffer another type of instability: sensitivity to input valuations (see Appendix \ref{app:sensitivity}).
%Experienced advertisers have found that \textit{duplication} is useful to counteract instability.
%Our results give a plausible explanation on why it works well in practice.

Instability differs in degree market-by-market and we will give more analysis in Appendix \ref{app:discussion_instability}.
To get a basic idea, consider a two-bidder market that is symmetric in the sense that goods appear in pair, of which one is valued $v_1$ and $v_2$ and the other is valued $v_2$ and $v_1$ by bidder 1 and 2, respectively (note that the key construction in the proof of APX-hardness shares a similar structure).
There is always an equilibrium where bidder 1 wins all goods, and one where bidder 2 wins all.
Depending on specific valuation profiles, there may also be many intermediate ones.
From a dynamic point of view, committing a higher multiplier would make the opponent pay more, and the bidder who quits the price war first would lower its multiplier to satisfy its ROI-constraint (and also the opponent's) but lose the market share.

In addition, the above prototypical example distinguishes two sources of instability: an intensely competitive landscape and the IC property.
First price auction also suffers high-sensitivity if it holds for  a large percentage of goods that the values of top bidders are extremely close.
However, in first price auction, competition is local and direct, i.e., bid or value perturbations only affect the auctions in which they happen.
But in second price auction, any fluctuation will propagate to the whole market through the runner-up-winner interdependence and the impact is more widespread and unpredictable.

\subsection{Interference in A/B Testing for Platforms}
% introduce a/b testing, sutva and interference
A/B testing is an indispensable tool to evaluate new technologies and assist business decisions.
In a typical setup, users in experiment are randomly assigned to either a treatment or a control variant (e.g., different reserve pricing strategies), and metrics are aggregated within each group to compare and see which variant is better.
The same idea can also be applied to randomize advertisers.
An ideal experiment requires the Stable Unit Treatment Value Assumption (SUTVA) to hold, which generally means that there should be no interference between the treatment and the control group.
Ad-side experiments (regardless of auction formats) clearly violate SUTVA since all ads compete for the same set of goods, and user-side violation (in non-auction scenarios) is also common in practice.
%Various designs have been proposed in literature to reduce bias and increase experimental power for different interference structures \citep{johari2020experimental, ha2020counterfactual, liu2020trustworthy, holtz2020reducing, karrer2021network}.
%\footnote{Offline-online inconsistency.
	%it is also a common experience in industry that a machine learning model performing better offline fails to bring a positive gain in online A/B testing.
	%The reason is surely complex, e.g., due to selection bias \citep{yuan2019improving} or overfitting \citep{zhou2021hybrid}.}
We show empirically that an unpredictable\footnote{It is known that A/B testing in two-sided markets suffers from \textit{cannibalization bias} \cite{blake2014marketplace,liu2020trustworthy}, which often enlarges the estimated advantages of the better variant.
	Such a bias may actually increase the experimental power since practitioners care more about whether a treatment is better, rather than how much better.
	In contrast, the interference introduced by IC is complex and it may lead to wrong decisions easily.}
bias exists broadly in naive implementations of both user-side and ad-side A/B testing in second price auction markets.
The bias comes from the fact that bidders' behaviors in  the (counterfactual) A/A, A/B and B/B tests are all different.
%\footnote{With ROI-constrained bidders, first price auction does not suffer from this type of interference. As for budget-constrained bidders, their behaviors in A/A and A/B tests are indeed different for both first and second price auction. But taking other properties like equilibrium uniqueness and tractable computation into consideration, first price auction should generally suffer less due to its convexity, and the interference is easier to deal with.}
Experiment details and more discussion can be found in Appendix \ref{app:ab_testing}, where we also propose a  simple approach to discerning biases and designing less-biased experiments.

\section{First Price Auction versus Second Price Auction: Within and Beyond Auto-bidding}
\label{sec:first_vs_second}

%\paragraph{First price auto-bidding equilibrium.}
With first price auction, no auto-bidding is needed and the platform simply allocates goods to bidders with the highest (tROI-discounted) valuations, which are also charged as payments.
The best possible revenue (see Appendix \ref{app:arbitrary_bid}) is naturally achieved.
Quasi-linear utility makes no difference for advertisers and it is a dominant strategy to report their tROIs truthfully. Even if they prefer spending less with the same acquired value, the incentive to deviate diminishes as the market becomes thicker.\footnote{This technically deviates from our model. Rigorous treatment is given in Appendix \ref{app:frugal}.
	%	Importantly, the equilibrium is still unique and always achieves the first-best welfare discounted by \textit{true} tROIs.
}
%If some advertiser does dominate all the auctions it participates and has a strong motivation to underbid, second price auction cannot generate more revenue either. The advertiser may submit tROI truthfully in second price auction, but it would never pay more than in first price auction. More revenue is only possible with extra mechanisms such as reserve price, irrelevant to auction formats.
First price auction does not provide advertisers an \textit{ex-post optimal} outcome.
However, advertisers should happily accept it since it is still \textit{fair/envy-free} as it relates closely to the classic \textit{competitive equilibrium} in Fisher markets (see Appendix \ref{app:fairness}).

% comparison of first and second in all regards mentioned above
Besides revenue-optimality, strategyproofness and fairness, first price auction also dominates in almost all the other aspects studied in this paper: (1) market outcome is unique; (2) computation is straightforward; (3) ROI-monotone; (4) no interference among sub-markets; (5) competition is direct and local; even if the competition is so intense that the utility becomes unstable, it can easily be smoothed by actively applying small perturbations or probabilistic allocations \cite{borgs2007dynamics}.
It shares with second price auction the problem of biases in ad-side A/B testing since it is rooted in the setup itself, irrelevant to auction formats.
Budget-constrained markets share similar results \cite{conitzer2021pacing}: (1) the equilibrium is unique; (2) computation is convex and tractable; (3) budget-monotone.
Since budget should still be paced, there remains interference among sub-markets and the competition is less direct than within ROI-constrained markets.
Nonetheless, with first price auction, the market is convex and more predictable, in contrast to the complex combinatorial structure of second price auction that is difficult to deal with.

In Appendix \ref{app:google_shift}, we extrapolate from our model to the case where advertisers may submit tROIs for several different sub-markets (e.g., via targeting in practice).
We argue that, if optional sub-markets are coarse-grained and advertisers do not have enough knowledge to differentiate them, the market could still be well captured by our model and first price auction mostly retains the upper hand.
As an application of our results, we also give our guess on why Google moves from second to first price auction for only Content, Video and Games, but not Search and Shopping. 
%Researchers have tried to explain the recent second-to-first trend in industry from the perspective of revenue by showing that, under different mathematical (but usually too idealized and restricted) models, choosing first price auction brings more revenue for the platform against their opponents (Paes Leme, Sivan, and Teng 2020; Despotakis, Ravi, and Sayedi 2021).
%We argue that it is the extent to which advertisers have access to individual auctions that determines the choice of auction formats: the greater a platform has a control over the auctions, the closer the market is to our model and the more advantages first price has over second price.

\section{Conclusion}
\label{sec:concluding_remarks}

% maths is typically extreme; provide intuitions for practitioners
% some metrics cannot be verified through experiments; some can (like stability)

%It is surly a complex problem in practice to choose a mechanism that allocates and prices ad slots among advertisers.
In this paper, we study a model that abstracts several most influential features of present auto-bidding markets.
%To show the effectiveness of  the reasoning and intuition we provide throughout the paper, we extrapolate our theory to the real world with a case study on Google's choice of auction formats.
%Since specific results are built upon various theoretical or empirical assumptions that are never perfect,
We try our best to pinpoint how the IC property contributes to every theoretical or empirical phenomenon such that readers could better extrapolate our results.
%We hope that researchers and practitioners could better extrapolate from the reasoning and intuition we provide to more complex markets in both theory and the real world.
We hope that our work could bring new perspectives to the community, and inspire practitioners to pay closer attention to the IC property and attain a better grasp of real-world markets.

%\appendix

%\section*{Ethical Statement}
%
%There are no ethical issues.

%\section*{Acknowledgments}

%% The file named.bst is a bibliography style file for BibTeX 0.99c
\bibliographystyle{named}
\bibliography{bibliography}

\begin{thebibliography}{}

\bibitem[\protect\citeauthoryear{Aggarwal \bgroup \em et al.\egroup
  }{2019}]{aggarwal2019autobidding}
Gagan Aggarwal, Ashwinkumar Badanidiyuru, and Aranyak Mehta.
\newblock Autobidding with constraints.
\newblock In {\em International Conference on Web and Internet Economics},
  pages 17--30. Springer, 2019.

\bibitem[\protect\citeauthoryear{Akbarpour and
  Li}{2020}]{akbarpour2020credible}
Mohammad Akbarpour and Shengwu Li.
\newblock Credible auctions: A trilemma.
\newblock {\em Econometrica}, 88(2):425--467, 2020.

\bibitem[\protect\citeauthoryear{Ausiello \bgroup \em et al.\egroup
  }{2012}]{ausiello2012complexity}
Giorgio Ausiello, Pierluigi Crescenzi, Giorgio Gambosi, Viggo Kann, Alberto
  Marchetti-Spaccamela, and Marco Protasi.
\newblock {\em Complexity and approximation: Combinatorial optimization
  problems and their approximability properties}.
\newblock Springer Science \& Business Media, 2012.

\bibitem[\protect\citeauthoryear{Babaioff \bgroup \em et al.\egroup
  }{2021}]{babaioff2021non}
Moshe Babaioff, Richard Cole, Jason Hartline, Nicole Immorlica, and Brendan
  Lucier.
\newblock Non-quasi-linear agents in quasi-linear mechanisms.
\newblock In {\em 12th Innovations in Theoretical Computer Science Conference
  (ITCS 2021)}. Schloss Dagstuhl-Leibniz-Zentrum f{\"u}r Informatik, 2021.

\bibitem[\protect\citeauthoryear{Balseiro and Gur}{2019}]{balseiro2019learning}
Santiago~R Balseiro and Yonatan Gur.
\newblock Learning in repeated auctions with budgets: Regret minimization and
  equilibrium.
\newblock {\em Management Science}, 65(9):3952--3968, 2019.

\bibitem[\protect\citeauthoryear{Balseiro \bgroup \em et al.\egroup
  }{2021a}]{balseiro2021robust}
Santiago Balseiro, Yuan Deng, Jieming Mao, Vahab Mirrokni, and Song Zuo.
\newblock Robust auction design in the auto-bidding world.
\newblock {\em Advances in Neural Information Processing Systems}, 34, 2021.

\bibitem[\protect\citeauthoryear{Balseiro \bgroup \em et al.\egroup
  }{2021b}]{balseiro2021budget}
Santiago Balseiro, Anthony Kim, Mohammad Mahdian, and Vahab Mirrokni.
\newblock Budget-management strategies in repeated auctions.
\newblock {\em Operations Research}, 2021.

\bibitem[\protect\citeauthoryear{Balseiro \bgroup \em et al.\egroup
  }{2021c}]{balseiro2021landscape}
Santiago~R Balseiro, Yuan Deng, Jieming Mao, Vahab~S Mirrokni, and Song Zuo.
\newblock The landscape of auto-bidding auctions: Value versus utility
  maximization.
\newblock In {\em Proceedings of the 22nd ACM Conference on Economics and
  Computation}, pages 132--133, 2021.

\bibitem[\protect\citeauthoryear{Balseiro \bgroup \em et al.\egroup
  }{2022}]{balseiro2022best}
Santiago~R Balseiro, Haihao Lu, and Vahab Mirrokni.
\newblock The best of many worlds: Dual mirror descent for online allocation
  problems.
\newblock {\em Operations Research}, 2022.

\bibitem[\protect\citeauthoryear{Blake and Coey}{2014}]{blake2014marketplace}
Thomas Blake and Dominic Coey.
\newblock Why marketplace experimentation is harder than it seems: The role of
  test-control interference.
\newblock In {\em Proceedings of the fifteenth ACM conference on Economics and
  computation}, pages 567--582, 2014.

\bibitem[\protect\citeauthoryear{Borgs \bgroup \em et al.\egroup
  }{2007}]{borgs2007dynamics}
Christian Borgs, Jennifer Chayes, Nicole Immorlica, Kamal Jain, Omid Etesami,
  and Mohammad Mahdian.
\newblock Dynamics of bid optimization in online advertisement auctions.
\newblock In {\em Proceedings of the 16th international conference on World
  Wide Web}, pages 531--540, 2007.

\bibitem[\protect\citeauthoryear{Charles \bgroup \em et al.\egroup
  }{2013}]{charles2013budget}
Denis Charles, Deeparnab Chakrabarty, Max Chickering, Nikhil~R Devanur, and Lei
  Wang.
\newblock Budget smoothing for internet ad auctions: a game theoretic approach.
\newblock In {\em Proceedings of the fourteenth ACM conference on Electronic
  commerce}, pages 163--180, 2013.

\bibitem[\protect\citeauthoryear{Chen \bgroup \em et al.\egroup
  }{2009}]{chen2009settling}
Xi~Chen, Xiaotie Deng, and Shang-Hua Teng.
\newblock Settling the complexity of computing two-player nash equilibria.
\newblock {\em Journal of the ACM (JACM)}, 56(3):1--57, 2009.

\bibitem[\protect\citeauthoryear{Chen \bgroup \em et al.\egroup
  }{2021a}]{chen2021complexity}
Xi~Chen, Christian Kroer, and Rachitesh Kumar.
\newblock The complexity of pacing for second-price auctions.
\newblock In {\em Proceedings of the 22nd ACM Conference on Economics and
  Computation}, pages 318--318, 2021.

\bibitem[\protect\citeauthoryear{Chen \bgroup \em et al.\egroup
  }{2021b}]{chen2021throttling}
Xi~Chen, Christian Kroer, and Rachitesh Kumar.
\newblock Throttling equilibria in auction markets.
\newblock {\em arXiv preprint arXiv:2107.10923}, 2021.

\bibitem[\protect\citeauthoryear{Conitzer \bgroup \em et al.\egroup
  }{2022a}]{conitzer2021pacing}
Vincent Conitzer, Christian Kroer, Debmalya Panigrahi, Okke Schrijvers,
  Nicolas~E Stier-Moses, Eric Sodomka, and Christopher~A Wilkens.
\newblock Pacing equilibrium in first price auction markets.
\newblock {\em Management Science}, 68(12):8515--8535, 2022.

\bibitem[\protect\citeauthoryear{Conitzer \bgroup \em et al.\egroup
  }{2022b}]{conitzer2021multiplicative}
Vincent Conitzer, Christian Kroer, Eric Sodomka, and Nicolas~E Stier-Moses.
\newblock Multiplicative pacing equilibria in auction markets.
\newblock {\em Operations Research}, 70(2):963--989, 2022.

\bibitem[\protect\citeauthoryear{Dasgupta and
  Tsui}{2004}]{dasgupta2004auctions}
Sudipto Dasgupta and Kevin Tsui.
\newblock Auctions with cross-shareholdings.
\newblock {\em Economic Theory}, 24(1):163--194, 2004.

\bibitem[\protect\citeauthoryear{Debreu}{1952}]{debreu1952social}
Gerard Debreu.
\newblock A social equilibrium existence theorem.
\newblock {\em Proceedings of the National Academy of Sciences},
  38(10):886--893, 1952.

\bibitem[\protect\citeauthoryear{Deng \bgroup \em et al.\egroup
  }{2021}]{deng2021towards}
Yuan Deng, Jieming Mao, Vahab Mirrokni, and Song Zuo.
\newblock Towards efficient auctions in an auto-bidding world.
\newblock In {\em Proceedings of the Web Conference 2021}, pages 3965--3973,
  2021.

\bibitem[\protect\citeauthoryear{Deng \bgroup \em et al.\egroup
  }{2022}]{deng2022efficiency}
Yuan Deng, Jieming Mao, Vahab Mirrokni, Hanrui Zhang, and Song Zuo.
\newblock Efficiency of the first-price auction in the autobidding world.
\newblock {\em arXiv preprint arXiv:2208.10650}, 2022.

\bibitem[\protect\citeauthoryear{Despotakis \bgroup \em et al.\egroup
  }{2021}]{despotakis2021first}
Stylianos Despotakis, R~Ravi, and Amin Sayedi.
\newblock First-price auctions in online display advertising.
\newblock {\em Journal of Marketing Research}, 58(5):888--907, 2021.

\bibitem[\protect\citeauthoryear{Fan}{1952}]{fan1952fixed}
Ky~Fan.
\newblock Fixed-point and minimax theorems in locally convex topological linear
  spaces.
\newblock {\em Proceedings of the National Academy of Sciences of the United
  States of America}, 38(2):121, 1952.

\bibitem[\protect\citeauthoryear{Filos-Ratsikas \bgroup \em et al.\egroup
  }{2021}]{filos2021complexity}
Aris Filos-Ratsikas, Yiannis Giannakopoulos, Alexandros Hollender, Philip
  Lazos, and Diogo Po{\c{c}}as.
\newblock On the complexity of equilibrium computation in first-price auctions.
\newblock In {\em Proceedings of the 22nd ACM Conference on Economics and
  Computation}, pages 454--476, 2021.

\bibitem[\protect\citeauthoryear{Gao and Kroer}{2020}]{gao2020first}
Yuan Gao and Christian Kroer.
\newblock First-order methods for large-scale market equilibrium computation.
\newblock {\em Advances in Neural Information Processing Systems},
  33:21738--21750, 2020.

\bibitem[\protect\citeauthoryear{Glicksberg}{1952}]{glicksberg1952further}
Irving~L Glicksberg.
\newblock A further generalization of the kakutani fixed point theorem, with
  application to nash equilibrium points.
\newblock {\em Proceedings of the American Mathematical Society},
  3(1):170--174, 1952.

\bibitem[\protect\citeauthoryear{Golrezaei \bgroup \em et al.\egroup
  }{2021}]{golrezaei2021auction}
Negin Golrezaei, Ilan Lobel, and Renato Paes~Leme.
\newblock Auction design for roi-constrained buyers.
\newblock In {\em Proceedings of the Web Conference 2021}, pages 3941--3952,
  2021.

\bibitem[\protect\citeauthoryear{Google}{2021}]{google2021shift}
Google.
\newblock Faqs about adsense moving to a first-price auction.
\newblock \url{https://support.google.com/adsense/answer/10858748#faqs}, 2021.
\newblock Accessed: 2024-05-08.

\bibitem[\protect\citeauthoryear{G{\"u}rtler and
  G{\"u}rtler}{2012}]{gurtler2012inequality}
Marc G{\"u}rtler and Oliver G{\"u}rtler.
\newblock Inequality aversion and externalities.
\newblock {\em Journal of Economic Behavior \& Organization}, 84(1):111--117,
  2012.

\bibitem[\protect\citeauthoryear{Ha-Thuc \bgroup \em et al.\egroup
  }{2020}]{ha2020counterfactual}
Viet Ha-Thuc, Avishek Dutta, Ren Mao, Matthew Wood, and Yunli Liu.
\newblock A counterfactual framework for seller-side a/b testing on
  marketplaces.
\newblock In {\em Proceedings of the 43rd International ACM SIGIR Conference on
  Research and Development in Information Retrieval}, pages 2288--2296, 2020.

\bibitem[\protect\citeauthoryear{Kagel and Levin}{2011}]{kagel2011auctions}
John~H Kagel and Dan Levin.
\newblock Auctions: A survey of experimental research, 1995-2010.
\newblock {\em Handbook of experimental economics}, 2:563--637, 2011.

\bibitem[\protect\citeauthoryear{Karande \bgroup \em et al.\egroup
  }{2013}]{karande2013optimizing}
Chinmay Karande, Aranyak Mehta, and Ramakrishnan Srikant.
\newblock Optimizing budget constrained spend in search advertising.
\newblock In {\em Proceedings of the sixth ACM international conference on Web
  search and data mining}, pages 697--706, 2013.

\bibitem[\protect\citeauthoryear{Kimbrough and
  Reiss}{2012}]{kimbrough2012measuring}
Erik~O Kimbrough and J~Philipp Reiss.
\newblock Measuring the distribution of spitefulness.
\newblock 2012.

\bibitem[\protect\citeauthoryear{Leme \bgroup \em et al.\egroup
  }{2012}]{leme2012sequential}
Renato~Paes Leme, Vasilis Syrgkanis, and {\'E}va Tardos.
\newblock Sequential auctions and externalities.
\newblock In {\em Proceedings of the twenty-third annual ACM-SIAM symposium on
  Discrete Algorithms}, pages 869--886. SIAM, 2012.

\bibitem[\protect\citeauthoryear{Liaw \bgroup \em et al.\egroup
  }{2022}]{liaw2022efficiency}
Christopher Liaw, Aranyak Mehta, and Andres Perlroth.
\newblock Efficiency of non-truthful auctions under auto-bidding.
\newblock {\em arXiv preprint arXiv:2207.03630}, 2022.

\bibitem[\protect\citeauthoryear{Liu \bgroup \em et al.\egroup
  }{2020}]{liu2020trustworthy}
Min Liu, Jialiang Mao, and Kang Kang.
\newblock Trustworthy online marketplace experimentation with budget-split
  design.
\newblock {\em arXiv preprint arXiv:2012.08724}, 2020.

\bibitem[\protect\citeauthoryear{Paes~Leme \bgroup \em et al.\egroup
  }{2020}]{paes2020competitive}
Renato Paes~Leme, Balasubramanian Sivan, and Yifeng Teng.
\newblock Why do competitive markets converge to first-price auctions?
\newblock In {\em Proceedings of The Web Conference 2020}, pages 596--605,
  2020.

\bibitem[\protect\citeauthoryear{Papadimitriou and
  Peng}{2021}]{papadimitriou2021public}
Christos Papadimitriou and Binghui Peng.
\newblock Public goods games in directed networks.
\newblock In {\em Proceedings of the 22nd ACM Conference on Economics and
  Computation}, pages 745--762, 2021.

\bibitem[\protect\citeauthoryear{Rabin}{1993}]{rabin1993incorporating}
Matthew Rabin.
\newblock Incorporating fairness into game theory and economics.
\newblock {\em The American economic review}, pages 1281--1302, 1993.

\bibitem[\protect\citeauthoryear{Rubinstein}{2018}]{rubinstein2018inapproximability}
Aviad Rubinstein.
\newblock Inapproximability of nash equilibrium.
\newblock {\em SIAM Journal on Computing}, 47(3):917--959, 2018.

\bibitem[\protect\citeauthoryear{Zhang \bgroup \em et al.\egroup
  }{2016}]{zhang2016feedback}
Weinan Zhang, Yifei Rong, Jun Wang, Tianchi Zhu, and Xiaofan Wang.
\newblock Feedback control of real-time display advertising.
\newblock In {\em Proceedings of the Ninth ACM International Conference on Web
  Search and Data Mining}, pages 407--416, 2016.

\end{thebibliography}

\clearpage
\newpage

\appendix

\section{Note on Numerical Experiments}

Numerical experiments are extensively used throughout the paper.
Appendix \ref{app:algorithms} will explain how the market instances are generated and how equilibria are computed.

\section{Notes on Terminology}
\label{app:terminologies}

\paragraph{ROAS, ROI and CPA.}
ROAS (Return-On-Ad-Spend) is generally unambiguous, measured as the ratio of the received value to the payment. Typically ROI means the ratio of the quasi-linear utility to the payment, but is sometimes used the same as ROAS. We will stick to the first usage in this paper.
Target CPA (Cost-Per-Action) can be interpreted as the amount of money the advertiser is willing to pay for every specified action taken by a user after interacting with the ad.
The three quantities are different forms of the same mathematical concept, but the ``directions'' of tCPA and tROI/tROAS are opposite: raising tCPA is equivalent to lowering tROI/tROAS by some appropriate factors.
In practice, tCPA is mainly used for objectives that are hard to directly value and discrete in nature, such as ad-clicks, app-downloads, user subscriptions, new customer visits, etc., while tROI/tROAS are used for those easily related to money and thus continuous, e.g., in-app purchases for mobile games and sales volumes for stores.

\paragraph{Conversion-rate.}
For discrete conversions, conversion-rate can be directly interpreted as the probability that a conversion happens after the user interacting with the ad. For conversions of continuous types (such as sales volume), the conversion-rate is the expected quantity of conversions received by the advertiser (e.g., the amount of money spent). Regardless, the value can always be decomposed into two parts, one of which is the relative magnitudes (predicted by the platform) among ad slots, and the other is a uniform factor (derived from the tCPA/tROI/tROAS submitted by the advertiser) that scales the relative magnitudes such that the resulting valuation would be comparable among advertisers.

\paragraph{Pacing.}
Pacing is initially used for budget-constrained campaigns where the job of an auto-bidder is to \textit{pace} the rate at which the budget is spent. %As a result, \citet{conitzer2021multiplicative, conitzer2021pacing} call their solution concept \textit{pacing equilibrium} for budget-constrained auto-bidders.
There is no budget for ROI-constrained bidders, but the strategy is so well-known and thus we stick to the name in this paper.

\section{More Discussion on Related Works}
\label{app:more_related_works}

\paragraph{Traditional RTB and auto-bidding literature.}
Traditionally RTB or auto-bidding algorithms
%(just to name a few: \cite{zhang2016feedback,ren2017bidding,kitts2017ad,maehara2018optimal,wu2018budget,yang2019bid,morishita2020online,tunuguntla2021near})
are mainly designed from a single advertiser's point of view.
One exception is the work by Aggarwal et al. \shortcite{aggarwal2019autobidding}, but their treatment of equilibrium existence is incomplete.
First, they overlooked the full allocation condition (see Definition \ref{def:autobidding_equilibrium}): what they prove exists is actually an “equilibrium up to tied goods” (see Appendix \ref{app:iterative}).
Second and more importantly, they assume a value distribution without point mass, which makes the problem continuous.
Note that, except for the discussion on equilibrium, the main body of their work  deals with the optimization problem faced by a single auto-bidder with \textit{discrete} valuations.
Discontinuity makes the fixed-point theorem inapplicable and the existence proof much harder.
They may already be aware of such difficulties and circumvent it by assuming continuity.
Also note that continuity (without further assumptions like independence) is only useful when establishing existence. Other than that, it is more convenient to rigorously discuss and study equilibrium properties with a general discrete valuation.

\paragraph{Pure Nash equilibrium and auto-bidding equilibrium.}
In our model, a pure Nash equilibrium may not always exist.
This is in contrast to the budget-constrained case \cite{conitzer2021multiplicative} where the equilibrium is a \textit{refinement} of PNEs (with a flexible tie-breaking rule that favors an equilibrium if one exits). This also justifies that the ROI-constraint deserves a separate treatment from the very beginning.

\paragraph{Externality.}
A simple example of externality is spitefulness (see, e.g., \cite{kimbrough2012measuring}), of which a bidder always prefers the winner spending more when it loses the auction. Other examples include inequality aversion \cite{gurtler2012inequality}, intention-based social preferences \cite{rabin1993incorporating}, cross-shareholdings between firms \cite{dasgupta2004auctions}, etc.
Leme et al. \shortcite{leme2012sequential} study a kind of internalized externality similar to ours in a sequential setting.

\paragraph{Rationale of the second-to-first trend in the industry.}
Researchers have tried to explain the recent trend where more and more platforms move from second to first price auction from the perspective of revenue.
They show that, under different mathematical (but usually too idealized and restricted) models, choosing first price auction brings more revenue for the platform against their opponents \cite{paes2020competitive,despotakis2021first}.
We take a more holistic approach and argue that it is advertisers' ability to differentiate sub-markets that determines the auction format: the more coarse-grained the optional sub-markets are, the closer the market is to our model and the more advantages first price has over second price auction.

\paragraph{Beyond multiplicative pacing.}
In this paper we focus on auto-bidders with the multiplicative pacing strategies.
Some of the recent auto-bidding literature \cite{liaw2022efficiency,deng2022efficiency} on first price auction still allow bidders to bid arbitrarily for individual auctions, which conforms more to the traditional auction design and RTB settings.
For budget-constrained auto-bidders, there is a strategy named \textit{throttling} that receives some attention in the literature \cite{chen2021throttling,balseiro2021budget}.
The behavior of throttling is quite different from multiplicative pacing: e.g., with throttling the utility of each bidder is naturally continuous w.r.t. the throttling parameter, while with multiplicative pacing the utility is discontinuous w.r.t. the multiplier.
There is no counterpart of throttling in ROI-constrained auto-bidding.
To the authors' knowledge, it is also much less used than multiplicative pacing in the industry  nowadays.

\section{Rationale behind the Definition of Auto-bidding Equilibrium}
\label{app:equilibrium_definition_rationale}

Recall that we consider a market where a set of bidders $N = \{1, \dots, n\}$ compete for a set of divisible goods $M = \{1, \dots, m\}$.
The tROIs of all bidders are set to zero, i.e., each bidder's spend should be no more than its acquired value.
ROI is calculated as the ratio of the quasi-linear utility to the payment, and ROAS = ROI + 1. Note that a bidder with valuations $\lambda v_{i, j}$'s and a tROI of $\lambda - 1$ will behave exactly the same as one with $v_{i, j}$'s and tROI zero.
Therefore the assumption of zero tROIs is WLOG.
Meanwhile, we implicitly use the tROI-discounted value to measure welfare. By assuming zero tROIs, no discount appears explicitly in calculation, but the zero-tROI assumption is WLOG only if the discount is always taken into account. If we change the tROI/tCPA of a bidder, we will never consider the welfare of the whole market.

We use $v_{i, j}$ to denote the value of bidder $i$ to good $j$.
For each good $j$, there is at least one bidder $i$ such that $v_{i, j} > 0$.
The ad platform simultaneously runs a single-item second price auction for every good.
Auto-bidders are restricted to apply multiplicative pacing strategies: the action space of bidder $i$ is the set of undominated multipliers $\alpha_i \in [1, +\infty)$,\footnote{In the definition of auto-bidding equilibrium and the proof of equilibrium existence (Appendix \ref{app:proof_existence}), there is an upper bound $A$ for multipliers. Theoretically, an upper bound makes the strategy space compact. It is easy to see that, for sufficiently large $A$, any equilibrium for  $\alpha \in [1, A]^n$ will be equivalent to an equilibrium without the cap.} and its bid for good $j$ is $\alpha_i v_{i, j}$.

If we specify a tie-breaking rule to make the allocation $x \in [0, 1]^{n \times m}$ (where $x_{i, j}$ is the fraction of good $j$ allocated to bidder $i$) uniquely determined by $\alpha$, and define bidder's utility as the ROI-constrained valuation:
\begin{displaymath}
	u_i(\alpha) =
	\left\{
	\begin{array}{ll}
		\sum_j  x_{i, j} v_{i, j}, & \text{if } \sum_{j} x_{i, j} p_j \leq \sum_{j} x_{i, j} v_{i, j};
		\\
		-\infty, & \text{otherwise;}
	\end{array}
	\right.
\end{displaymath}
then the pure Nash equilibrium (PNE) of this normal form game seems to be a reasonable solution concept.
However, no matter how ties are broken, there exists instances where a PNE does not exist.
On the other hand, it turns out that there always exists some pure strategy  profile that constitutes a steady-state of the market.
Take the market given in Table \ref{tab:no_pne} with 2 bidders and 2 goods as an example.
\begin{table}[h]
	\centering
	\begin{tabular}{c cc}
		\toprule
		valuation &  good $1$ & good $2$   
		\\
		\midrule
		bidder $1$ & $1$ & $1$ 
		\\
		bidder $2$ & $0$ & $3$
		\\
		\bottomrule
	\end{tabular}
	\caption{A market without PNE.}
	\label{tab:no_pne}
\end{table}
Suppose for now that we are in a dynamic setting where 2 simultaneous auctions are repeated each round.
Bidders are restricted to multiplicative pacing within each round, but allowed to adjust multipliers across time.
Since bidder $1$ can always win good $1$ for free, it has an incentive to oscillate its bid $b_{1, 2}$ around $3$ to win good $2$ \textit{sometimes} with a price around $3$.
In response, bidder $2$ will keep its bid $b_{2, 2}$ at $3$, as bidding more risks violating its ROI-constraint.
The resulting long-term average allocation should be $x_{1, 1} = 1, x_{2, 1} = 0, x_{1, 2} = x_{2, 2} = 0.5$ with prices $p_1 = 0, p_2 = 3$.
However, no PNE can achieve this outcome regardless of tie-breaking rules, since with $b_{1, 2} = 3$, bidder $2$ always wants to raise its bid above $3$ to win good 2 in whole.
On the other hand, $(\alpha_1, \alpha_2) = (3, 1)$ does constitute a stable state, and there is no need to introduce mixed strategies.

To circumvent the non-existence of PNE, we directly define our solution concept, the auto-bidding equilibrium (Definition \ref{def:autobidding_equilibrium}), which reasonably characterizes the steady-state of the market and always exists.
For the market in Table \ref{tab:no_pne}, the unique auto-bidding equilibrium is $\alpha_1 = 3, \alpha_2 = 1$ and $x_{1, 1} = 1, x_{2, 1} = 0, x_{1, 2} = x_{2, 2} = 0.5$, exactly as anticipated.

\section{Reserve Prices and Additive Boosts}
\label{app:reserve_and_boosts}

We will use reserve prices and additive boosts in Section \ref{subsec:ppad}, and Appendix \ref{subsec:seller_competition}, \ref{app:arbitrary_bid} and \ref{app:ab_testing}.
With reserve prices, the full allocation condition only needs to hold if the highest bid is strictly larger than the reserve price. With additive boosts, bidders are ranked by the boosted bid $b_{i, j} + c_{i, j}$ ($= \alpha_i v_{i, j} + c_{i, j}$ with multiplicative pacing) where $c_{i, j}$ are constants chosen by the seller, independent of $b$. The winner is charged with the second highest boosted bid minus its own boost (or equivalently, the minimum non-boosted bid $b_{i, j}$ for bidder $i$ to win).

\section{Proof of Theorem \ref{thm:existence} (Existence of Auto-bidding Equilibrium)}
\label{app:proof_existence}

The high-level picture of the proof follows the methodology commonly employed in existence results built on the theorem by Debreu \shortcite{debreu1952social}, Fan \shortcite{fan1952fixed} amd Glicksberg \shortcite{glicksberg1952further}, i.e., the original discontinuous game is approximated by a series of smoothed instances whose PNEs are guaranteed to exist.
Here we have two sources of discontinuities: one lies in the payment and allocation resulting from the auction rule and discrete valuations; another lies in the utility due to the hard ROI-constraints.

A standard approach to smooth the former is to divide goods among the set of bidders whose bids are close enough to the highest bid, such that the share of allocation and payment is continuous with respect to the multiplier.
For ROI-constraints, note that bidder $i$ has an incentive to raise bid and win more low-ROI goods if its quasi-linear utility is positive (and otherwise it would lower bid and give up some low-ROI goods).
In the proof we will call a negative quasi-linear utility \textit{debt}.
To tackle utility discontinuity, besides the acquired value, we will add a negative term to the utility for each unit of debt a bidder owes.
When the bidder has a negative debt, the coefficient of this term is set small enough to maintain its original incentive for acquiring more value.
Otherwise, the coefficient will be made sufficiently large to impose the hard ROI-constraint in a continuous way.

ROI-constraints bring two extra difficulties.
First, unlike payment, sometimes the debt will decrease with respect to the multiplier, since a bidder may win a fraction of some good with a positive marginal-ROI due to the smoothed allocation.
This is overcome by lower bounding the strategy space slightly above 1, such that all goods with positive marginal-ROI would be allocated fully even at the minimum multiplier, and winning any extra goods would thus bring a positive debt.
However, the lower bound comes with the second difficulty: a bidder may violate the ROI-constraint at the minimum multiplier, which puts the limiting point at the same risk.
The solution is to add an infinitesimal cold-start fund to guarantee that the total debt is negative when the multiplier lies on the lower bound.

The proof proceeds as follows.
\begin{definition}
	For $\epsilon > 0$ and $H > 0$, an \textbf{$(\epsilon, H)$-smoothed game} is a normal form game over an auto-bidding market where the set of pure strategies for each bidder $i$ is the set of multipliers $\alpha_i \in \bigbrackets{ 1 + \frac{2\epsilon}{v_i^*} , A}$,  where $v_i^* = \min_{j : v_{i, j} > 0} v_{i, j}$.
	Bidder $i$'s bid for good $j$ is still $b_{i, j} = \alpha_i v_{i, j}$, but the auction rule and bidders' utility functions are modified as follows:
	
	\textbf{Allocation and payment rule:}
	for each good $j$, consider the highest bid $b_j^* = \max_k \alpha_k v_{k, j}$. Let $S_j = \{i: \alpha_i v_{i, j} \geq b_j^* - \epsilon \}$ be the set of bidders close to the first price winner for $j$. Then for $i \in S_j$, the allocation $x$ is given by
	\begin{displaymath}
		x_{i, j} = \frac{\alpha_i v_{i, j} - (b_j^* - \epsilon)}{\sum_{k \in S_j} [\alpha_k v_{k, j} - (b_j^* - \epsilon)] },
	\end{displaymath}
	and $p_{i, j}$ (bidder $i$'s payment) is the highest bid on good $j$ excluding  bidder $i$. For other bidders, $x_{i, j} = 0$.
	
	\textbf{Additional artificial good:}
	each bidder will additionally receive a quantity $\alpha_i$ of an artificial good (with unlimited supply) worth $2\epsilon$ per unit, and afford a debt of $\epsilon$ per unit. This results in a profit of $\alpha_i \epsilon$ if the bidder is out of debt, and a large cost otherwise.
	
	\textbf{Cold-start fund:}
	each bidder starts with a fund of $\epsilon^{1/2}$.
	We will call $- \sum_j v_{i, j} x_{i, j} + \sum_j p_{i, j} x_{i, j}$ the \textit{real} debt, and $\alpha_i \epsilon - \epsilon^{1/2}$ the \textit{artificial} debt.
	For sufficiently small $\epsilon$, the artificial debt is negative.
	%$\epsilon \bigparen{ \alpha_i - \frac{1}{\epsilon^{1/2}} } < 0$
	
	\textbf{Utility:} $u_i(\alpha) = ( \sum_j v_{i, j} x_{i, j} - \sum_j p_{i, j} x_{i, j} - \alpha_i \epsilon + \epsilon^{1/2}   ) + 2\alpha_i \epsilon + A  \sum_{j} v_{i, j} x_{i, j}$
	if $\sum_j v_{i, j} x_{i, j} - \sum_j p_{i, j} x_{i, j} - \alpha_i \epsilon + \epsilon^{1/2}  \geq 0$, otherwise
	$u_i(\alpha) = H \bigparen{\sum_j v_{i, j} x_{i, j} - \sum_j p_{i, j} x_{i, j} - \alpha_i \epsilon + \epsilon^{1/2}  } + 2\alpha_i \epsilon + A \sum_{j} v_{i, j} x_{i, j}$.
	
\end{definition}

\begin{lemma}
	For any $(\epsilon, H)$-smoothed game with $H > \max\left\{ 2, \frac{A \max v_{i, j}}{  \epsilon} \right\}$,  a PNE always exists.
\end{lemma}
\begin{proof}
	The classic theorem \cite{debreu1952social,fan1952fixed,glicksberg1952further} assures us that a PNE always exists if the game satisfies the following three conditions, which we will verify as follows:
	
	\textbf{Compact and convex strategy space.} $\alpha_i \in \bigbrackets{ 1 + \frac{2\epsilon}{v_i^*} , A}$.
	
	\textbf{Continuity of utility w.r.t. the strategy profile.} $b_j^*$ is continuous in $\alpha$. $x_{i, j}$ and $p_{i, j}$ are continuous in $\alpha$ (and $b_j^*$) (in particular, bidder $i$ who is just barely in $S_j$ with $\alpha_i v_{i, j} = b_j^* - \epsilon$ receives zero allocation).
	And the utility is continuous in $\alpha, x$ and $p$ (in particular, when $\sum_j v_{i, j} x_{i, j} - \sum_j p_{i, j} x_{i, j} - \alpha_i \epsilon + \epsilon^{1/2}  = 0$, the expressions coincide at $2 \alpha_i \epsilon + A \sum_{j} v_{i, j} x_{i, j}$).
	
	\textbf{Quasiconcavity of utility in the bidder's own strategy.} Fix $\alpha_{-i}$, let $t$ be the infimum of the set of $\alpha_i$ such that $\sum_j v_{i, j} x_{i, j} - \sum_j p_{i, j} x_{i, j} - \alpha_i \epsilon + \epsilon^{1/2}  \leq 0$ (if no such value exists, set $t = A$).
	
	For $\alpha_i < t$, rearrange $u_i$ as follows:
	\begin{displaymath}
		u_i(\alpha) =
		\sum_{j} v_{i, j} x_{i, j}
		+
		\alpha_i \epsilon
		+
		\sum_j \bigparen{A v_{i, j} - p_{i, j}} x_{i, j}
		+
		\epsilon^{1/2}.
	\end{displaymath}
	$u_i$ is strictly increasing in $\alpha_i$, since $p_{i, j}$ is fixed, $x_{i, j}$ is increasing in $\alpha_i$, and $A v_{i, j} \geq \alpha_i v_{i, j} \geq p_{i, j}$ if $x_{i, j} > 0$.
	
	On the other hand,
	if $p_{i, j} < v_{i, j} + \epsilon \leq \bigparen{1 + \frac{2\epsilon}{v_i^*}} v_{i, j} - \epsilon$, then $x_{i, j} = 1$. Therefore, $i$'s real debt $ \sum_j (p_{i, j} - v_{i, j}) x_{i, j}$ will always increase in $x_{i, j}$, which makes the total debt $\sum_j v_{i, j} x_{i, j} - \sum_j p_{i, j} x_{i, j} - \alpha_i \epsilon + \epsilon^{1/2}  \leq 0$ for all $\alpha_i \geq t$. 
	So for $\alpha_i \geq t$, we can rearrange $u_i$ as
	\begin{align*}
		u_i(\alpha)
		=&
		- H \sum_j (p_{i, j} - v_{i, j}) x_{i, j}
		- (H - 2) \epsilon \alpha_i
		\\
		&+ \sum_j A v_{i, j} x_{i, j}
		+
		H\epsilon^{1/2}.
	\end{align*}
	The second term is strictly decreasing in $\alpha_i$ for $H > 2$.
	Since $p_{i, j} \geq v_{i, j} + \epsilon$ for any newly acquired good $j$, the first term decreases the utility at a rate of at least $H \epsilon$ in terms of $x_{i, j}$.
	And if
	$H > \frac{A v_{i, j}}{\epsilon}$,
	the first and the third term combined will also decrease in $x_{i, j}$, and thus $u_i$ is strictly decreasing in $\alpha_i$ when $\alpha_i \geq t$.
\end{proof}

\begin{proof}[Proof of Theorem \ref{thm:existence}]
	Consider a sequence of smoothed games defined by $(\epsilon^k, H^k)$ satisfying $H^k > \max\left\{ 2, \frac{A \max v_{i, j}}{  \epsilon^k} \right\}$, and $\lim_k \epsilon^k = 0$.
	Since the set of pacing multipliers, allocations and payments is compact, we can pick a converging sequence of equilibria of these games $\{ \alpha_i^k, x_{i, j}^k, p_{i, j}^k \}_{i \in N, j \in M} \rightarrow \{\alpha_i^*, x_{i, j}^*, p_{i, j}^*\}_{i \in N, j \in M}$.
	We should check that $(\alpha^*, x^*)$ forms an auto-bidding equilibrium.
	
	\textbf{Goods go to the highest bidders.}
	If $x_{i, j}^* > 0$, then for sufficiently large $k$, $\alpha_i^k v_{i, j} \geq \max_{i'} \alpha_{i'}^k v_{i', j} - \epsilon^k$. Since $\lim_k \epsilon^k = 0$, we have $\alpha_i^* v_{i, j} \geq \max_{i'} \alpha_{i'}^* v_{i', j}$.

	\textbf{Winner pays the second price.}
	$p_{i, j}^k$ is the highest bid among bidders excluding $i$ itself, so it converges to the highest bid among other bidders at the limit point.
	
	\textbf{Full allocations of goods.}
	For each $k$ and $j$, $\sum_{i} x_{i, j}^k = 1$.
	
	\textbf{ROI-feasible.}
	Suppose that for some bidder $i$, $\sum_j p_j^* x_{i, j}^* > \sum_j v_{i, j} x_{i, j}^*$.
	Then there exists $\delta > 0$ such that for any $K$, we can find $k > K$ with $\sum_j p_{i, j}^k x_{i, j}^k - \sum_j v_{i, j} x_{i, j}^k > \delta$.
	For sufficiently large $k$, the artificial debt $\alpha_i \epsilon - \epsilon^{1/2}$ will be less than $\delta$, which results in a strictly positive total debt.
	However, by bidding $1 + \frac{2\epsilon}{v_i^*}$, the real debt is at most $m \epsilon \bigparen{\frac{2 \max_i v_{i, j}}{v_i^*} + 1}$, which is less than the negative artificial debt $\epsilon \bigparen{\frac{1}{\epsilon^{1/2}} - \alpha_i}$ for sufficiently small $\epsilon$.
	Thus by the strict quasiconcavity of the utility, in this case $i$ could choose an $\alpha_i^k$ such that the total debt would be zero and its utility would be strictly higher,  a contradiction.
	
	\textbf{Maximal pacing.}
	%		\textbf{(1, non-binding ads)}
	If $\sum_{j} p_j^* x_{i, j}^* < \sum_{j} v_{i, j} x_{i, j}^*$, then there exists some $K$ such that for any $k > K$, $\sum_{j} p_{i, j}^k x_{i, j}^k < \sum_{j} v_{i, j} x_{i, j}^k $.
	By the strict quasiconcavity of the utility, $\alpha_i^k = A$, so $\alpha_i^* = A$.
	%		\textbf{(2, non-winning ads)} 
	%		Suppose that $\alpha_i^* v_i r_{i, j} < \max_{i' \neq i} \alpha_{i'}^* v_{i'} r_{i', j}$ for all $j$'s, then for sufficiently large $k$, $\alpha_i^k v_i r_{i, j} <  \max_{i' \neq i} \alpha_{i'}^k v_{i'} r_{i', j} - \epsilon^k$.
	%		However, by the strict quasiconcavity of the utility, raising $\alpha_i^k v_i r_{i, j}$ to $\max_{i'} \alpha_{i' \neq i}^k v_{i'} r_{i', j} - \epsilon^k$ for some $j$ while maintaining $x_{i, j'} = 0$ for all $j'$ will give $i$ a strictly better utility (note that if $x_{i, j'} = 0, \forall j'$, then for sufficiently large $k$, $v_i \sum_j x_{i, j} r_{i, j} - \sum_j p_{i, j} x_{i, j} - \alpha_i \epsilon + \epsilon^{1/2}  =  \epsilon \bigparen{ \frac{1}{\epsilon^{1/2}} - \alpha_i } > 0)$.
\end{proof}

\section{Proof of Theorem \ref{thm:ppad_hardness} (PPAD-hardness of Finding Any Equilibrium)}
\label{app:proof_ppad_hardness}

%\begin{repeattheorem}[Theorem \ref{thm:ppad_hardness}.]
%	Finding an $(\eta, \delta)$-approximate auto-bidding equilibrium is PPAD-hard for some constant $\eta, \delta > 0$.
%\end{repeattheorem}

We will prove the result by reducing from the problem of finding an $\epsilon$-approximate equilibrium of a \textit{threshold game}.
A threshold game is defined over a directed graph $G = (V, E)$ with a threshold parameter $t \in (\epsilon, 1 - \epsilon)$.
Each vertex $u$ is a player with action space $y_u \in [0, 1]$.
An action profile forms an $\epsilon$-approximate equilibrium if for every vertex $u$ and the set of its in-neighbors $N_u$, it satisfies that
\begin{displaymath}
	y_u \in \left\{
	\begin{array}{ll}
		{[0, \epsilon]}, & \text{if } \sum_{w \in N_u} y_w > t + \epsilon; \\
		{[1 - \epsilon, 1]}, & \text{if } \sum_{w \in N_u} y_w < t - \epsilon; \\
		{[0, 1]}, & \text{otherwise}.
	\end{array}
	\right.
\end{displaymath}
The problem is known  to be PPAD-complete for some constant $\epsilon > 0$, any value of $t \in (\epsilon, 1 - \epsilon)$, and any graph where the in-degree and the out-degree of each vertex is at most 3 \cite{papadimitriou2021public}.
In the reduction we will choose $t = 1/2$.
We will first reduce an instance of the threshold game to an auto-bidding market with reserve prices (see definitions in Section \ref{sec:auto_bidding_equilibrium}).
Later we will show how to remove reserve prices while maintaining the correctness of the reduction.

\begin{table*}[t]
	\centering
	\begin{tabular}{c ccc c c}
		\toprule
		& $(w_1, u)$ & $(w_2, u)$  & $(w_3, u)$  & $\underline{u}$ & $\bar{u}$
		\\
		\midrule
		in-neighbor $w_1$ & (1 - $\eta$)/14
		\\
		in-neighbor $w_2$ & & (1 - $\eta$)/14
		\\
		in-neighbor $w_3$ & & & (1 - $\eta$)/14
		\\
		bidder $u$ & 1/3 & 1/3 & 1/3 & (1 - $\eta$)/2 & 1/(4 - 4$\eta$)
		\\
		reserve prices & $(1 - \eta)/7$  & $(1 - \eta)/7$ & $(1 - \eta)/7$ & 1 &  1
		\\
		\bottomrule
	\end{tabular}
	\caption{Valuations related to a vertex $u$.}
	\label{tab:ppad_main_construction}
\end{table*}

The construction takes $(\eta, \delta)$ as inputs (for now just treat them as two numbers).
The market consists of $|V|$ bidders and $5|V|$ goods.
For each vertex $u \in V$, there are a vertex bidder, a lower bound good $\underline{u}$, an upper bound good $\bar{u}$, and three incoming edge goods.
We associate each incoming edge $(w, u) \in E$ of vertex $u$ to one of its corresponding incoming edge goods.
For simplicity, we will name a vertex bidder or edge good by its corresponding vertex or edge.
The meaning will be clear from the context and we will not refer to an edge good that is not associated to any edge (this happens when the vertex has less than three in-neighbors).
All lower and upper bound goods have a reserve price $1$.
All edge goods have a reserve price $(1 - \eta)/7$.
%If an edge item has no associated edge, set a reserve price $(1 - \eta)/7$ for it.
Only the corresponding vertex bidder $u$ is interested in bound good $\underline{u}$ and $\bar{u}$ with $v_{u, \underline{u}} = (1 - \eta)/2$ and $v_{u, \bar{u}} = 1/(4 - 4\eta)$.
For each edge good $(w, u)$, only bidder $w$ and $u$ value it positively with $v_{u, (w, u)} = 1 / 3$ and $v_{w, (w, u)} = (1  - \eta) / 14$.
Table \ref{tab:ppad_main_construction} summarizes the construction of valuations.

\begin{lemma}
	If $(\alpha, x)$ is an $(\eta, \delta)$-approximate auto-bidding equilibrium of the market constructed with parameter $(\eta, \delta)$ and $\eta, \delta > 0$ are sufficiently small, then $\alpha_{u} \in [2, 4], \forall u \in V$, and edge goods will be sold fully to their corresponding head bidders.
\end{lemma}
\begin{proof}
	In this proof, all inequalities hold strictly when $\eta = 0$ and $\delta = 0$.
	By continuity, there exist sufficiently small $\eta > 0$ and $\delta > 0$ that maintain the strictness of these inequalities.
	
	If $\alpha_u > 4$, it will win good $\underline{u}$ and $\bar{u}$ in whole and pay 2 for them, but the total value of all the goods in which bidder $u$ is interested is only $\frac{7}{4} - \frac{\eta}{2} + \frac{\eta}{4(1 - \eta)} < \frac{2}{1 + \delta}$, violating the ROI-feasible condition.
	
	Given that all multipliers are upper bounded by 4, bidder $w$ could bid at most $\frac{2(1 - \eta)}{7} <  \frac{1}{3} (1 - \eta)$ to the outgoing edge good $(w, u)$, so good $(w, u)$ will be sold fully to bidder $u$.
	If $\alpha_u < 2$, bidder $u$ wins and only wins the incoming edge goods of total value $1$, but pays at most $\frac{ 3 \cdot 2 (1 - \eta)}{7} < 1 - \delta$, violating the maximal pacing condition.
\end{proof}

\begin{lemma} \label{lemma:ppad_main_construction}
	Given  an $(\eta, \delta)$-approximate auto-bidding equilibrium $(\alpha, x)$ of the market constructed with parameter $(\eta, \delta)$, 	construct an action profile $y$ of the threshold game by setting $y_u = \frac{1}{2} \alpha_u  - 1 \in [0, 1]$ for every $u \in V$.
	Then $y$ is an $\epsilon$-approximate equilibrium of the threshold game if $\eta$ and $\delta$ is sufficiently small and $\epsilon = \frac{7(3 - \eta)\delta }{ 2(1 - \eta)}$.
\end{lemma}
\begin{proof}
	Consider three different cases of the sum of in-neighbors' actions for each vertex $u$.
	\begin{itemize}
		\item $\sum_{w \in N_u} y_w > \frac{1}{2} + \epsilon$.
		Then $\sum_{w \in N_u} \alpha_u = \sum_{w \in N_u} 2(y_w + 1) > 2 |N_u| + 1 + 2 \epsilon$.
		Bidder $u$ gets a value of 1 from incoming edge goods and pays
		\begin{align*}
			&\frac{1 - \eta}{14} \bigparen{ 2 (3 - |N_u|) + \sum_{w \in N_u} \alpha_w}
			\\
			> &
			\frac{(1  - \eta)(7 + 2 \epsilon)}{14}
			\\
			=&
			\bigparen{\frac{1}{2} - \frac{\eta}{2}} + \frac{\epsilon (1 - \eta)}{7}.
		\end{align*}
		If $\alpha_u > 2 + 2 \epsilon$, bidder $u$ will win good $\underline{u}$ in whole.
		The total value of incoming edge goods and $\underline{u}$ is $\frac{3}{2} - \frac{\eta}{2}$, but the payment is strictly larger than $\frac{3}{2} - \frac{\eta}{2} + \frac{\epsilon (1 - \eta)}{7} = \bigparen{\frac{3}{2} - \frac{\eta}{2}} (1 + \delta)$.
		(Winning $\bar{u}$ could only deviate from the target ROI further.)
		Therefore $\alpha_u \leq 2 + 2 \epsilon$, i.e., $y_u \leq \epsilon$.
		\item $\sum_{w \in N_u} y_w < \frac{1}{2}- \epsilon$.
		Now we have $\sum_{w \in N_u} \alpha_u < 2|N_u| + 1 - 2\epsilon$.
		If $\alpha_u < 4 - 2\epsilon$, then $\frac{4 - 2\epsilon}{4(1 - \eta)} < 1$, and bidder $u$ only wins incoming edge goods and $\underline{u}$.
		In this case the total payment is strictly lower than $\frac{3}{2} - \frac{\eta}{2} - \frac{\epsilon (1 - \eta)}{7} = \bigparen{\frac{3}{2} - \frac{\eta}{2}}  (1 - \delta)$.
		\item $\sum_{w \in N_u} y_w \in \left[\frac{1}{2} - \epsilon, \frac{1}{2} + \epsilon\right]$.
		From the above two case analyses, we can learn that the ratio of payment to valuation from buying incoming edge good and $\underline{u}$ falls in the feasible range $\bigbrackets{1 - \delta, 1 + \delta}$.
		Therefore $\alpha_u$ could be any number in $[2, 4]$.
	\end{itemize}
\end{proof}

We move on to replace a reserve price of value $r$ with the gadget shown in Table \ref{tab:reserve_price_gadget}.
$j_0$ is the good for which we want to set a reserve price.
Note that, besides auxiliary bidders, there may be other bidders who are interested in $j_0$ but not shown in the table.
In Lemma \ref{lemma:auxiliary_range}, \ref{lemma:perturbed_equilibrium} and Corollary \ref{corollary:reserve_price}, the parameter $(\eta, \delta)$ is \textit{generic} and different from the one used in constructing the auto-bidding market.

\begin{table*}[t]
	\centering
	\begin{tabular}{c ccc}
		\toprule
		&  auxiliary good $j_1$ & auxiliary good $j_2$   &  target good $j_0$
		\\
		\midrule
		auxiliary bidder $i_1$ & 0 & $2r$ & $r$
		\\
		auxiliary bidder $i_2$ & $r/2$ & $r$ & $r/2$
		\\
		\bottomrule
	\end{tabular}
	\caption{Valuations of a reserve price gadget of value $r$ for good $j_0$. The gadget resembles the example in Appendix \ref{app:equilibrium_definition_rationale} that shows the non-existence of PNE and motivates us to define our own solution concept.}
	\label{tab:reserve_price_gadget}
\end{table*}

\begin{lemma} \label{lemma:auxiliary_range}
	At an $(\eta, \delta)$-approximate auto-bidding equilibrium with a reserve price gadget as in Table \ref{tab:reserve_price_gadget}, the price $p_{j_2}$ of good $j_2$ is within the range $[2r(1 - \eta), 2r(1 + \delta)]$, and we have $\alpha_{i_1} \in \left[1, \frac{1 + \delta}{(1 - \eta)^2}\right], \alpha_{i_2} \in \left[2(1 - \eta), \frac{2(1 + \delta)}{1 - \eta}\right]$.
\end{lemma}
\begin{proof}
	If $p_{j_2} < 2r(1 - \eta)$, then $\alpha_{i_2} < 2(1 - \eta)$ and $i_2$ wins only $j_1$ with value $r/2$ but pays nothing, violating the maximal pacing condition.
	The argument also shows $\alpha_{i_2} \geq 2(1 - \eta)\alpha_{i_1} \geq 2(1 - \eta)$.
	
	If $p_{j_2} > 2r(1 + \delta)$, then $i_1$'s ROI-constraint would be violated if it won $j_2$, and $j_2$ could only be fully sold to $i_2$.
	If so, however, $i_2$ would pay strictly more than $2r(1 + \delta)$ on $j_1$ and $j_2$, but they only generate a value of $3r/2$, violating its ROI feasible condition.
	(In this case, $p_{j_0} \geq p_{j_2} > \max(v_{i_1, j_0}, v_{i_2, j_0})$. Winning $j_0$ will only deviate the target ROI further for both $i_1$ and $i_2$.)
	
	If $\alpha_{i_2} > \frac{2(1 + \delta)}{1 - \eta}$, $i_1$ still does not want to win anything and the same argument can be applied as above.
	
	The range of $\alpha_{i_1}$ is deduced from the range of $\alpha_{i_2}$ by the inequality $\alpha_{i_2} \geq 2(1 - \eta)\alpha_{i_1}$.
\end{proof}

\begin{lemma} \label{lemma:perturbed_equilibrium}
	Suppose that $(\alpha, x)$ is an $(\eta, \delta)$-approximate auto-bidding equilibrium.
	If $\alpha'$ satisfies that, for each bidder $i$, $\frac{\alpha'_i}{\alpha_i} \in [a, b]$ for some constant $a \in (0, 1)$ and $b > 1$, then $(\alpha', x)$ is an $(\eta', \delta')$-approximate auto-bidding equilibrium with $\eta' = 1 - \frac{a(1 - \eta)}{b}$ and $\delta' = \max\{b(1 + \delta) - 1, 1 - a(1 - \delta)\}$.
\end{lemma}
\begin{proof}
	With multipliers $\alpha$, the ratio of the highest bid to the lowest bid to win is at most $1 / (1 - \eta)$.
	With $\alpha'$, the ratio is at most $\frac{b}{a(1 - \eta)} = \frac{1}{1 - \eta'}$.
	
	Similarly, the payment should satisfy $b(1 + \delta) \leq 1 + \delta'$ and $a(1 - \delta) \geq 1 - \delta'$.
\end{proof}

\begin{corollary} \label{corollary:reserve_price}
	Suppose that $(\alpha, x)$ is an $(\eta, \delta)$-approximate auto-bidding equilibrium of an auto-bidding market with reserve price gadgets.
	Let $\alpha'$ be identical to $\alpha$ except that, for each good $j_0$ equipped with a reserve price gadget of value $r$, $\alpha'_{i_1} = 1$ and $\alpha'_{i_2} = 2$ where $i_1$ and $i_2$ are the auxiliary bidders associated with $j_0$ as defined in Table \ref{tab:reserve_price_gadget}.
	Then $(\alpha', x)$ is an $(\eta', \delta')$-approximate auto-bidding equilibrium where $\eta' = 1 - \frac{(1 - \eta)^4}{1 + \delta}$ and $\delta' = \max\bigbraces{\frac{\eta + \delta}{1 - \eta}, 1 - \frac{(1 - \delta)(1 - \eta)^2}{1 + \delta}}$.
	
	Furthermore, when restricted to non-auxiliary bidders and goods, $(\alpha', x)$ is still an $(\eta', \delta')$-approximate auto-bidding equilibrium with the corresponding reserves prices.
\end{corollary}

Now we can put all components together to finish the proof.
\begin{proof}[Proof of Theorem \ref{thm:ppad_hardness}.]
	Suppose that we are given an approximation parameter $\epsilon$ and a directed graph $(V, E)$ where the  in-degree and the out-degree of any vertex are at most 3.
	We can construct an auto-bidding market (with reserve price gadgets) as in Table \ref{tab:ppad_main_construction} and \ref{tab:reserve_price_gadget} with parameter $(\eta_1, \delta_1)$.
	
	Suppose that $(\alpha, x)$ is an $(\eta_2, \delta_2)$-approximate auto-bidding equilibrium of the market.
	We can construct $(\alpha', x)$ for the corresponding auto-bidding market with reserve prices (and without reserve price gadgets) as in Corollary \ref{corollary:reserve_price} such that (1) it is an $(\eta_3, \delta_3)$-approximate auto-bidding equilibrium; (2) $\eta_3$ and $\delta_3$ go to zero as $\eta_2$ and $\delta_2$ go to zero; (3) reserve prices are set properly.
	
	Construct an action profile $y$ of the original threshold game by setting $y_u = \frac{1}{2} \alpha'_u  - 1 \in [0, 1]$ for every $u \in V$.
	By choosing sufficiently small $(\eta_2, \delta_2)$ in the previous step, we can made $\eta_3 \leq \eta_1, \delta_3 \leq \delta_1$.
	Then Lemma \ref{lemma:ppad_main_construction} can be applied to show that $y$ is an $\epsilon'$-approximate equilibrium where $\epsilon = \epsilon(\eta_1, \delta_1)$ goes to zero as $\eta_1$ and $\delta_1$ go to zero.
	
	All construction can be done in poly-time, and with sufficiently small $\eta_1, \delta_1 > 0$, the realized approximation ratio $\epsilon'$ could be made smaller than the target $\epsilon$.
\end{proof}

\paragraph{Further remarks.}
The reduction also maintains the sparse structure of the original threshold game in the sense that each bidder is only interested in at most 8 goods and each good is only valued positively by at most 3 buyers (the reserve price gadget can be removed if an edge good is associated with an edge).
%(2) The result gives the worst case complexity and holds for small approximation parameters.
%Possibly surprisingly, in our numerical experiments on realistic datasets, the algorithm in Section \ref{subsec:iterative} seems to converge relatively well.

\section{Proof of Theorem \ref{thm:complexity} (APX-hardness of Finding Optimal Equilibrium)}
\label{app:proof_apx_hardness}

\begin{proof}[Proof of Theorem \ref{thm:complexity}]
	We prove the result by an L-reduction from the well-known APX-complete problem: MAX-3SAT-3 \cite{ausiello2012complexity}.
	An instance of 3SAT consists of $n$ variables $\{x_i\}$ and $m$ clauses of the form $(l_1 \lor l_2 \lor l_3)$ where $l_k \in \{\pm x_i \}, k = 1, 2, 3$ is a literal of some variable.
	The optimization problem MAX-3SAT is to find the assignment $\{0, 1\}^n$ to $\{x_i\}$ that maximizes the number of satisfied clauses.
	MAX-3SAT-3 is a further restriction of MAX-3SAT where each variable appears at most 3 times.
	
	We reduce an arbitrary MAX-3SAT-3 instance with $n$ variables and $m$ clauses to the following auto-bidding market.
	%	Set $T$ equal to $0.075  + m$.
	For every variable $x_j$, create bidders $1^{x_j}, 2^{x_j}$ and goods $1^{x_j}, 2^{x_j}$, with $v_{1, 1} = v_{2, 2} = 0.05/n, v_{1, 2} = v_{2, 1} = 0.025/n$.
	For every clause $c$, create bidders $3^c, 4^c, 5^c$ and goods $3^c,  4^c$ with $v_{3, 3} = 0.5, v_{3, 4} = 0.1, v_{4, 4} = v_{5, 4} = 0.5$.
	For ``clause'' goods, we associate bidder $1^{x_j}$ with the literal $+x_j$, bidder $2^{x_j}$ with the literal $-x_j$, and the value $v_{1, 3} = 0.1$ if $+x_j$ occurs in the clause $c$, $v_{2, 3} = 0.1$ if $-x_j$ occurs in the clause.
	Valuations not mentioned are set zero.
	
	The following three results will be used to connect an auto-bidding equilibrium with an assignment of variables.
	\begin{enumerate}
		\item For any $x_j$, at equilibrium, $\min(\alpha_{1^{x_j}}, \alpha_{2^{x_j}}) \leq 2$, otherwise the total price of goods $1^{x_j}$ and $2^{x_j}$ will exceed $0.1/n$, the welfare generated is at most $0.1/n$, and bidder $1^{x_j}$ and $2^{x_j}$ cannot compensate this deficit by winning other goods.
		\item Bidder $1^{x_j}$ and bidder $2^{x_j}$ cannot win good $3^c$. This is because that winning any fraction of good $3^c$ requires an $\alpha \geq 5$, and since $\min(\alpha_{1^{x_j}}, \alpha_{2^{x_j}}) \leq 2$, one of $1^{x_j}$ and $2^{x_j}$ will win both good $1^{x_j}$ and $2^{x_j}$ with a total price at least $0.075/n$, which is already binding the ROI-constraint, and winning more low-ROI goods can only violate it.
		This also implies that, at equilibrium, if one of $\alpha_{1^{x_j}}$ and $\alpha_{2^{x_j}}$ is larger than 2, the other will be 1.
		\item The price of good $4^c$ is 0.5 at any equilibrium, otherwise either bidder $4^c$ or $5^c$ will violate their ROI-constraints, or bidder $3^c$ will win good $4^c$ in whole, in which case its value is at most 0.6 but it pays strictly more.
	\end{enumerate}

	Given an equilibrium of the market, if the price of good $3^c$ lies in $(0.2, 0.5)$ for some $c$, then the runner-up $i^{x_j}$ has a multiplier larger than $2$, which means it will win both good $1^{x_j}$ and $2^{x_j}$. If the price of good $3^c$ is less than $0.5$, bidder $3^c$ will raise bid to win a positive fraction of good $4^c$. In this case, we can increase the multiplier of $i^{x_j}$ to 0.5 to increase the revenue.
	Therefore we can compute another equilibrium (in poly-time) where the price of good $3^c$ is either $0.5$ or no larger than $0.2$ for all $c$.
	Furthermore, if the price of good $3^c$ lies in $(0.1, 0.2]$, then for every variable $x_j$ appearing in $c$ we have $\alpha_{1^{x_j}} \leq 2$ and $\alpha_{2^{x_j}} \leq 2$.
	Fix such a $x_j$ and let $C(+x_j)$ be the set of clauses which $+x_j$ appears in and have a price no larger than 0.2.
	Define $C(-x_j)$ similarly.
	Suppose WLOG that $|C(+x_j)| \geq |C(-x_j)| \geq 1$.
	Then by setting $\alpha_{1^{x_j}} = 5$ and $\alpha_{2^{x_j}} = 1$, the revenue increases by at least $0.3 |C(+x_j)| - 0.1 |C(-x_j)| - 0.05/n > 0$.
	Hence we can compute yet another equilibrium where the price of good $3^c$ is either 0.1 or 0.5 for every clause $c$.
	
	Now construct an assignment of the MAX-3SAT-3 instance by setting to TRUE those literals whose associated literal buyer has a multiplier strictly larger than 2.
	If the multipliers of two literal buyers are both not larger than 2, assign TRUE/FALSE arbitrarily.
	The assignment is feasible since  $\min(\alpha_{1^{x_j}}, \alpha_{2^{x_j}}) \leq 2$.
	
	Let OPT(A) and OPT(B) be the optimal objective value of the MAX-3SAT-3 instance and the optimal revenue at some equilibrium in the constructed market, respectively.
	Also suppose that the revenue of the equilibrium used to construct the assignment is $T$, and the corresponding assignment satisfies $m'$ clauses.
	To show that the above two-way construction forms an L-reduction, we need to show that for some constants $\beta, \gamma > 0$: (1) $\text{OPT(B)} \leq \beta \cdot \text{OPT(A)}$; (2) $\text{OPT(A)} - m' \leq \gamma \bigparen{\text{OPT(B)} - T}$.
	
	For condition (1), the optimal welfare of any equilibrium is at most $2 \times \frac{0.05}{n} \times n + 2 \times 0.5 \times m = 0.1 + m < 4n$ (where every good is allocated to the bidder with the highest valuation), and the (optimal) revenue is always bounded by welfare.
	For the MAX-3SAT-3 instance, at least $m/2 \geq n/6$ clauses can be satisfied in the optimal assignment.
	So $\text{OPT(B)} \leq 24 \cdot \text{OPT(A)}$.
	
	For condition (2),
	 if $\text{OPT(A)} = m'$, the inequality holds for any $\gamma$.
	Below we assume $\text{OPT(A)} - m' \geq 1$.
	By the construction, if clause $c$ is not satisfied, the price of good $3^c$ is 0.1.
	And if it is satisfied, the price is 0.5.
	Then we have
	\begin{displaymath}
		T \leq 0.1 + 0.5m + 0.1(m - m') + 0.5m'.
	\end{displaymath}
	On the other hand, from the optimal assignment we can construct an equilibrium by setting the multiplier of TRUE literal to 5 and FALSE literal to 1 (other ones can be easily set).
	As a result,
	\begin{displaymath}
		\text{OPT(B)} \geq 0.5m + 0.5 \text{OPT(A)} + 0.1 \bigparen{ m - \text{OPT(A)} }.
	\end{displaymath}
	Combining both and we have:
	\begin{align*}
		\text{OPT(B)} - T
		\geq &\,
		0.4 \bigparen{ \text{OPT(A)} - m'} - 0.1
		\\
		\geq &\,
		0.3 \bigparen{ \text{OPT(A)} - m'},
	\end{align*}
	which concludes the L-reduction.
	
	By our construction, the revenue is always equal to welfare. Thus finding the welfare-optimal equilibrium is also APX-hard.
\end{proof}

\section{Algorithms and Experimental Market Instances}
\label{app:algorithms}

We develop two algorithms to which we will refer throughout the paper as the \textit{MIBLP} and the \textit{iterative method}, respectively.
Before giving technical details, we first discuss their usages and implications in this paper.

% introduce instance types
Market instances used in our experiments are categorized into  \textit{realistic} and \textit{synthetic}.
Realistic instances are simply the \textit{raw} bidding data taken from a large real-world ad platform.
Appendix \ref{app:market_instances} will explain how each class of synthetic instances is generated.
%Sizes of realistic instances differ case-by-case and will be given as they are used.
Equilibrium conditions will be programmatically checked for every solution we get.
For synthetic instances, all constraints should be approximately satisfied.
If the iterative method does not converge within the specified time, the result will be discarded.
%We never find a clear cycling pattern and we are not sure whether it will always converge given enough time.
For realistic instances, the algorithm stops if the metrics of top bidders (in particular, those reported in the text) have approximately converged.

MIBLP is able to accurately compute the objective value (e.g., the maximum and minimum utility an advertiser can achieve at equilibrium).
If the iterative method is used, the objective can only be aggregated over a subset of all the equilibria, but it should \textit{not} affect our results since it provides a \textit{lower bound}, which is enough for our purpose (e.g., the utility instability issue can only be \textit{worse} if there are more unexplored equilibria in Appendix \ref{sec:instability}).

The convergence of the iterative method on realistic instances is surprisingly consistent: for every realistic instance tested in our pre-tests (not shown in this paper), we try multiple parameter configurations for the algorithm and it always returns essentially the same equilibrium (up to negligible differences).
Based on this, we only run the iterative method once for each realistic instance.
The performance does \textit{not} mean that: (1) the instance has a unique equilibrium; (2) the algorithm will work well on realistic value distribution.
The latter deserves some further discussion that concerns a limitation of our realistic dataset.

The platform from which our data are taken employs a \textit{multi-stage} auction mechanism,
where valuations will only be predicted at the last stage for a small subset of bidders who are selected through previous stages as the most competitive ones.\footnote{The multi-stage mechanism is employed almost everywhere in the industry, since it is simply impossible to predict the valuations and sort the bids of all advertisers for every single auction. For stages other than the last, there are different kinds of scores to rank bidders, but they cannot be directly mapped to the valuation.}
As a result, our realistic instances consist only of valuations  of bidders surviving to the last stage of each auction.
In practice, the filtering procedure will take the current multipliers $\alpha$ into account (possibly in some obscure way), and the set of bidders selected into the last stage will change dynamically.
It is impractical to simulate this in our paper since we only have access to historic data and it will be very costly to acquire counterfactual information that is not already produced.
As shown by markets constructed in the hardness proofs, filtering out some advertisers for each auction may simplify the problem, but the set of equilibria may change as well.\footnote{Another evidence is that, as we clustered the realistic data into markets of different sizes (experiments not shown in this paper), the iterative method does converge more and more slowly as the market size increases. But keep in mind that clustered markets differ from both the original and each other in structure.}
%Hence the performance of the iterative method on realistic instances might not carry over to the real-world.
Nonetheless, our realistic instances are valid markets in themselves, and we believe that this limitation would not significantly affect our results and the reasoning behind.

\subsection{MIBLP Formulation}
\label{subsec:miblp}

% TODO: emphasize the improvement over MIP

We formulate the problem of computing an equilibrium using Mixed-Integer BiLinear Programming (MIBLP), which allows  general (non-convex) \textit{quadratic} terms in constraints and objectives.
Previously, Conitzer et al. \shortcite{conitzer2021multiplicative} avoid non-linearity by only including payments into decision variables, while allocations could only be deduced from the solution after the solver terminates.
As a result, welfare cannot appear in the objectives of their formulation.
In our case, quadratic terms arise inevitably since ROI-constraints must be expressed using (variables of) both allocations and payments.
It turns out that our formulation does not induce higher time-complexity, but has the advantage of directly optimizing a broader set of objectives.

The constraints of the formulation are given below.
\begin{displaymath}
	\arraycolsep=3pt
	\begin{array}{ll}
		% line 1
		\sum_{i} s_{i, j} = p_j, \forall j &
		w_{i, j} \leq d_{i, j}, \forall i, j \\
		% line 2
		s_{i, j} \leq M d_{i, j}, \forall i, j &
		\sum_i w_{i, j} = 1, \forall j \\
		% line 3
		h_j \geq \alpha_i v_{i, j}, \forall i, j &
		\sum_i r_{i, j} = 1, \forall j \\
		% line 4
		h_j \leq \alpha_i v_{i, j} + (1 - d_{i, j}) M, \forall i, j  &
		r_{i, j} + w_{i, j} \leq 1, \forall i, j \\
		% line 5
		p_j \geq \alpha_i v_{i, j} - w_{i, j} M, \forall i, j
		&
		v_{i, j} s_{i, j} = p_j u_{i, j}, \forall i, j \\
		% line 6
		p_j \leq \alpha_i v_{i, j} + (1 - r_{i, j})M, \forall i, j 
		&
		\sum_j s_{i, j} = \sum_j u_{i, j}, \forall i
		\\
	\end{array}
\end{displaymath}
Table \ref{tab:miblp_variables} explains the meaning of each symbol. All the constraints but the last one are meant to encode the auction rules, which are essentially equivalent to those by Conitzer et al. \shortcite{conitzer2021multiplicative} except that we can directly use variables denoting allocations. We refer readers to their work for a proof of correctness.
\begin{table}[h]
	\begin{center}
		\begin{tabular}{llp{3.5cm}}
			\toprule
			Variable & Range & Meaning
			\\ \midrule
			$\alpha_i$ & $[1, +\infty)$ & multipliers \\
			$s_{i, j}$ & $[0, +\infty)$ & $i$'s payment on good $j$ \\
			$p_j$ & $[0, +\infty)$ & second price of good $j$ \\
			$h_j $ & $[0, +\infty)$ & first price of good $j$ \\
			$v_{i, j}$ & constant & valuations \\
			$u_{i, j} $ & $[0, +\infty)$ & the fraction of valuation won by $i$ from good $j$ \\
			$d_{i, j} $ & $\{0, 1\}$ & 1 if $i$ is among the first-price winners \\
			$w_{i, j} $ & $\{0, 1\}$ & 1 if $i$ is the designated first-price setter%\footnote{If a bidder is designated as the price setter, its bid $\alpha_i v_{i, j}$ will be used as the first or second price in the formulation. This does not mean that it will win a positive fraction of the good.}
			\\
			$r_{i, j} $ & $\{0, 1\}$ & 1 if $i$ is the designated second-price setter \\
			$M$ & a large constant &  used for the Big M method \\
			\bottomrule
		\end{tabular}
	\end{center}
	\caption{Symbols used in the MIBLP formulation.}
	\label{tab:miblp_variables}
\end{table}%

To enforce the ROI-feasible and the maximal pacing condition, one way is to introduce an additional binary variable to denote whether bidder $i$ has a binding ROI-constraint.
However, in this paper, we will always know a priori that no bidder could dominate all the auctions it participates.
With this presumption, all ROI-constraints should be binding, and we only need a linear constraint (the last one) to do the job.

The formulation itself is an \textit{exact} characterization of the computational task, but practical solvers are inevitably discrete and the solution is technically always an approximate equilibrium.
Nonetheless, the numeric errors in our experiments are small enough to be ignored and we will simply take each returned solution as an exact equilibrium.

\subsection{Iterative Method}
\label{app:iterative}

The iterative method represents a generic class of algorithms where each bidder updates its own multiplier to better respond to the current (or recent) bidding profile.
At each round $t$, goods are allocated and priced according to the current multipliers $\alpha_{i, t}, i \in \{1, \dots, n\}$.
The general multiplier updating rule is
\begin{displaymath}
	\alpha_{i, t + 1} \leftarrow \alpha_{i, t} + d_{i, t} s_{i, t},
\end{displaymath}
where $d_{i, t} \in \{-1, 0, 1\}$ is a better response direction, i.e., to raise bid if the ROI-constraint is satisfied and non-binding, and lower bid if violated, and $s_{i, t}$ is the step size.
The algorithm only returns a single equilibrium, and we will try different parameter configurations (e.g., initial multipliers) if multiple equilibria are to be explored.

Depending on the choice of $d_{i, t}$, the algorithm, if converges, returns different solution concepts.

(1)
$d_{i, t} = \text{sign}\left(\sum_j x_{i, j, t} (v_{i, j} - p_{j, t})\right)$.
If converges, the result is not an approximate equilibrium, but \textit{equilibrium up to tied goods}:
for each $i$, let $T_i$ be the set of goods of which $i$ is a tied winner.
Then it satisfies that $\sum_{j \notin T_i} x_{i, j} (v_{i, j} - p_{j}) \geq 0$ and $\sum_{j \notin T_i} x_{i, j} (v_{i, j} - p_{j}) + \sum_{j \in T_i} (v_{i, j} - p_{j}) \leq 0$.
From an \textit{individual} bidder's perspective, if it were able to freely choose any fraction of goods in $T_j$, $\alpha_i$ would be a best response.
For large instances, the result of a small number of auctions becomes insignificant, and the solution will form an approximate equilibrium then.

(2)
$
d_{i, t} = \frac{1}{2} \text{sign}\left(\sum_j x_{i, j, t} (v_{i, j} - p_{j, t})\right) + \frac{1}{2} \text{sign}\left(\sum_{\tau \in [t - k, t]}\sum_j x_{i, j, \tau} (v_{i, j} - p_{j, \tau})\right).
$
The idea is to take the recent $k$ iterations into consideration so as to reach coordination with others.
If converges, the result is an approximate equilibrium, and the accuracy tends to increase for larger $k$, but with a slower convergence rate.
Note that, to get an approximate equilibrium, the second term alone is sufficient.
The first term is meant to stabilize the dynamics.

The step size acts as a combination of learning rate and gradient.
We do some tuning for different instances. A typical choice is the logarithm of ROAS or
\begin{align*}
	s_{i, t} \propto & \frac{1}{t} \cdot \min \left( \abs{ \sum_j x_{i, j, t} (v_{i, j} - p_{j, t}) }, \right.
	\\
	& \qquad \qquad \left. \abs{ \sum_{\tau \in [t - k, t]}\sum_j x_{i, j, \tau} (v_{i, j} - p_{j, \tau}) } \right).
\end{align*}

If the iterative method is used, different equilibria will be found by running the algorithm for $2n+2$ different configurations of initial multipliers constructed as follows.
Let $\underline{\alpha}$ and $\overline{\alpha}$ be a smaller and a larger multiplier respectively. 
The bidders start with one of the following: (1) all $\underline{\alpha}$; (2) all $\overline{\alpha}$; (3) all $\underline{\alpha}$ except one $\overline{\alpha}$; (4) all $\overline{\alpha}$ except one $\underline{\alpha}$.
For all experiments $\underline{\alpha} = 1.0$.
$\overline{\alpha}$ is computed for each combination of instance type, distribution and market sizes by finding an equilibrium with all initial multipliers set to 1.0 and taking the median of equilibrium multipliers.

\subsection{Synthetic Instance Generation}
\label{app:market_instances}
Synthetic instances of each category are generated as follows.
\begin{itemize}
	\item \textit{Complete}:
	every bidder is interested in every good;
	each $v_{i, j}$ is drawn i.i.d. from $\text{Uniform}[0, 1]$ or $\text{Lognormal}(0, 1)$.
	\item \textit{Sampled}:
	constructed from a complete instance by dropping some $v_{i, j}$ to zero;
	for each good, the number $k$ of bidders to be dropped is drawn uniformly random from $\{0, 1, \dots, n-2\}$ (at least two bidders remain), and then $k$ bidders are dropped uniformly random from $\{1, \dots, n\}$.
	\item \textit{Correlated}:
	for each good $j$, $\mu_j$ is drawn from $\text{Uniform}[0, 1]$, and for each bidder, $v_{i, j}$ is drawn from $\text{Normal}(\mu_j, \sigma^2)$ and truncated to non-negative;
	correlated instances are always sampled.
\end{itemize}

% data limitation: multi-staged system
% realistic instance: our algorithm converges consistently; does not mean only one exists; regardless since one is enough in this paper; aggregate converge slow
% in general, with size grows, converges slower

\section{Exploitability by Bidders and Sellers}

Properties studied in this section are mostly presented using \textit{examples} to demonstrate the counter-intuitiveness of equilibrium.
In theory, a counterexample is enough to disprove a general statement, but it varies from market to market how often these (undesirable) events occur.
The non-monotonicity example in Appendix \ref{app:non_monotonicity} is randomly generated.
It is not hard to encounter one in both synthetic and realistic instances, but not frequent either.
For seller competition in Appendix \ref{subsec:seller_competition} or the deviation from multiplicative pacing in Appendix \ref{app:arbitrary_bid}, they may seldom be exploited \textit{on purpose} in reality, but could be triggered by innocent-looking behaviors.
We hope that practitioners could benefit from these new perspectives we provide when analyzing complex real-world scenarios.

\subsection{Non-monotonicity}
\label{app:non_monotonicity}

\begin{figure}[t]
	\centering
	\includegraphics[width=\columnwidth]{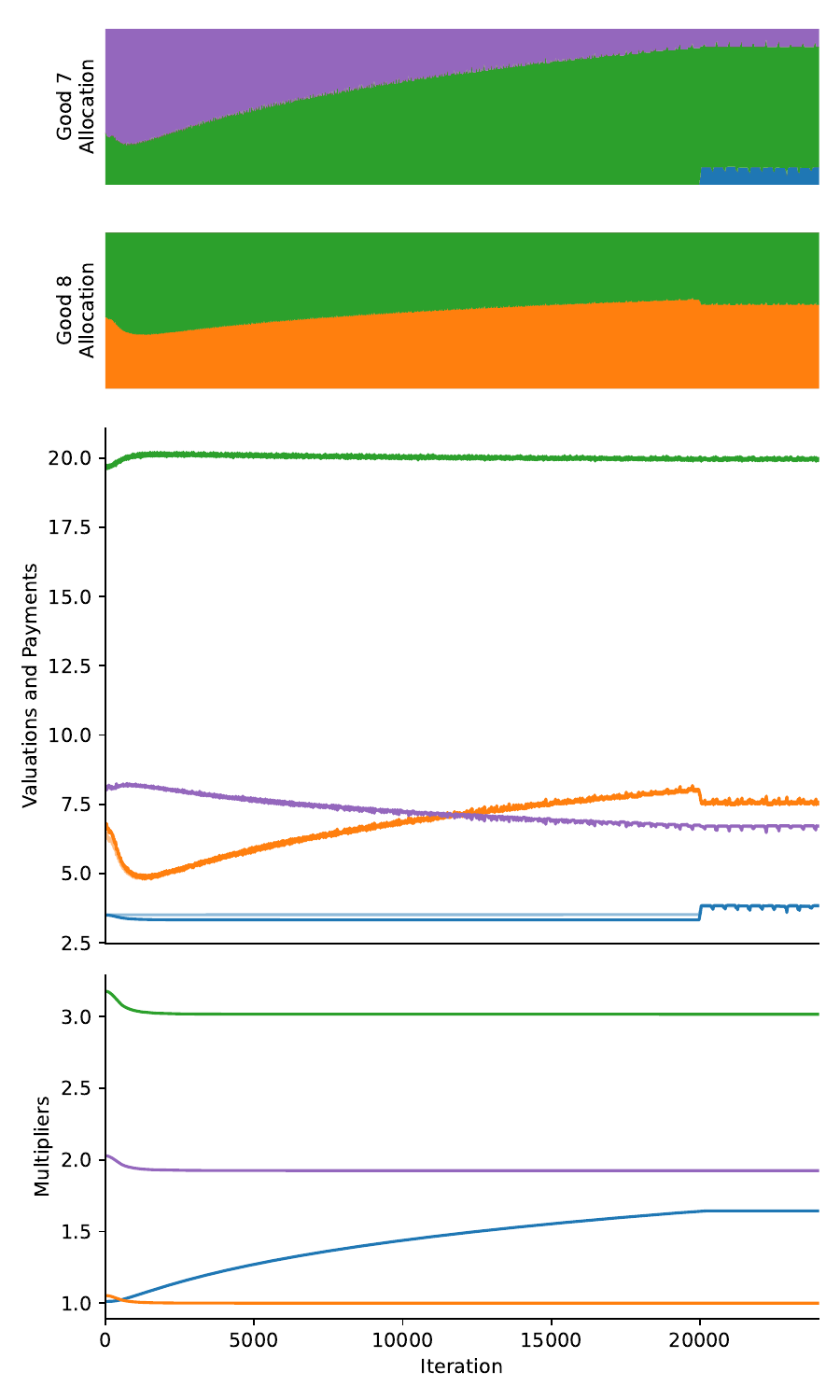}
	\caption{An instance of non-monotonicity.} \label{fig:non_monotonicity}
\end{figure}

%In a single-item second price auction, raising bids will surely increase the probability of winning.
If an auto-bidder lowers its multiplier, it will definitely win no more value \textit{immediately}.
However, other bidders will react to this change and the market will shift to a new equilibrium.
Below we demonstrate such an example of transition in detail.

The market consists of 5 bidders and 10 goods.\footnote{As the instance is generated randomly, we hide numbers that may instead confuse readers like valuations and equilibrium profiles.}
Figure \ref{fig:non_monotonicity} shows the dynamics after the change of tROI (technically tCPA) of some bidder at the old equilibrium, with time going by from left to right.
Bidders are denoted by their colors in the figure.
Only 4 bidders are plotted since the rest one wins nothing in both the old and the new equilibrium.
During the transition, only the allocation of good 7 and 8 changes, which are depicted in the top two subplots: the vertical length (relative to the total) denotes the average fraction of auctions won by the bidder within a moving window pivoted at each time step.
In the third subplot, we simultaneously draw the curves for valuations and payments.
For bidders other than Blue, their ROI-constraints are always (almost) binding and thus there is only a single curve for each color.
A twin curve is visible for Blue, whose payment (real line) is slightly less than valuation (translucent line) before around 20000 time steps, and as a result, Blue keeps raising its multiplier as shown in the fourth subplot.

At time 0, our manipulator  {\color{orange} Orange} changes its tCPA to 0.95 of the original, resulting in a valuation profile $v'$ such that $v_{Orange, j}' = 0.95 v_{Orange, j}, \forall j$.
Each bidder applies the iterative method (Appendix \ref{app:iterative}) to optimize its utility.
Recall that, if the current ROI (aggregated over a moving window of recent rounds) is too high, the algorithm will lower its multiplier, and vice versa.
To make the transition smooth, $\alpha_{Orange}$ is divided by $0.95$ right after the tCPA modification, which keeps the fine-grained bids of Orange unchanged at the moment.
$\alpha_{Orange} = 1$ at the old equilibrium and good 8 is the only good of which Orange wins a positive fraction.

Orange initiates the transition by lowering $\alpha_{Orange}$, since its ROI-constraint is now violated after the update of tCPA.
However, Green does not want to win good 8 completely, otherwise its ROI-constraint will also be violated.
So Green lowers its multiplier as well, which further triggers the same behavior for Purple.
As a result, Orange, Green and Purple reach an almost perfect coordination where the  multiplicative ratios among their multipliers remain nearly constant all the way through the transition.
This can also be seen from the fact that Purple and Green always tie for good 7, and Orange and Green always tie for good 8.

Blue's allocation keeps unchanged for the first 20000 time steps.
But it pays less due to the lowered second prices set by the other three and therefore tries to win more goods by gradually raising its multiplier.
During the process, Purple and Green pay more for goods whose second prices are set by Blue, and thus they have to give up goods of which they are one of the tied winners (these goods have the lowest marginal ROIs):
Purple gives up good 7 to Green, and Green wins more good 7 but loses good 8 to balance the deficit, which contributes to the success of the manipulation of Orange.
In the end, Blue takes a fraction of good 7 away from Green to bind its ROI-constraint, and Green compensates this by taking a fraction of good  8 from Orange.
Nonetheless, Orange still benefits from lowering its tCPA.

%Due to equilibrium multiplicity, it is hard to strictly define monotonicity or incentive-compatibility for  the advertiser game.
%But it is evident from the above example that, in real-world markets, advertisers have the opportunities to benefit from strategic behaviors.

\subsection{Competition among Sellers}
\label{subsec:seller_competition}

Besides bidder manipulation, there may also be (unintended) competition among sellers. Nowadays a large platform allows a campaign to advertise across its many ad networks, such as search, video, app store, etc. They are typically managed by different teams who are only responsible for their own metrics.
%Theoretically it should not be unexpected (since ad networks are intertwined through bidders’ behaviors) but
It is not hard to see that, when bidders adjust bids based on the overall performance across the platform (a common practice in the industry), there might be interference among different ad networks.
This may call for a closer coordination among ad networks and drive up the cost of cross-team communication.

To see this quantitatively, we demonstrate an example showing that it is possible to increase the revenue of one’s own ad network while lowering the efficiency of the platform.
We take bidding data of two ad networks within the same platform during the same time period as the valuation profile, which consists of 59330 bidders and 222791 goods.
%The iterative method (Appendix \ref{subsec:iterative}) is used to find the equilibrium.
We will apply different reserve pricing strategies\footnote{Similar to the strategy proposed by Balseiro et al. \shortcite{balseiro2021robust}. The strategy details are not important: we just need to apply \textit{some} treatment.} to \textit{only one} of the ad networks, and observe the treatment effects on \textit{both}.

\begin{figure*}[h]
	\centering
	\includegraphics[width=0.99\columnwidth]{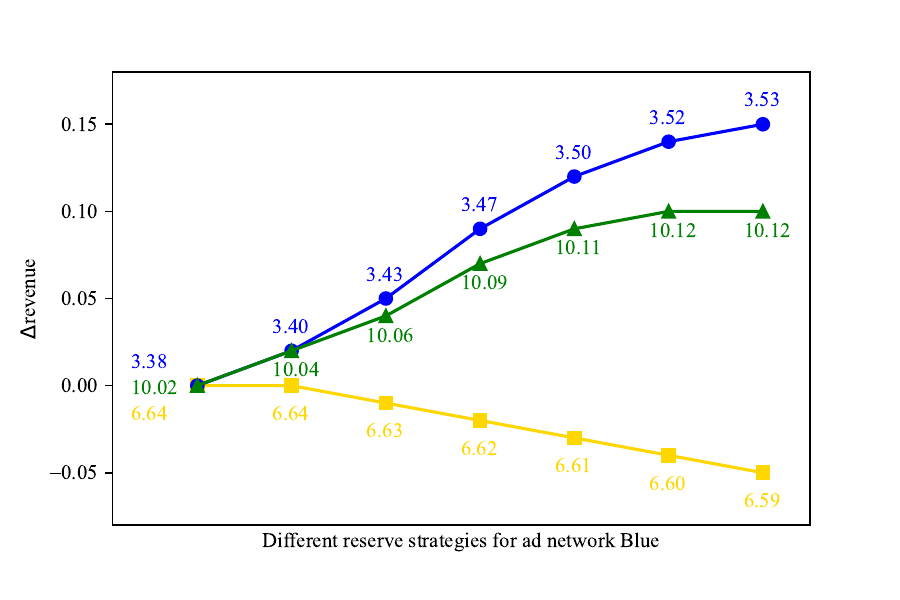}
	\includegraphics[width=0.99\columnwidth]{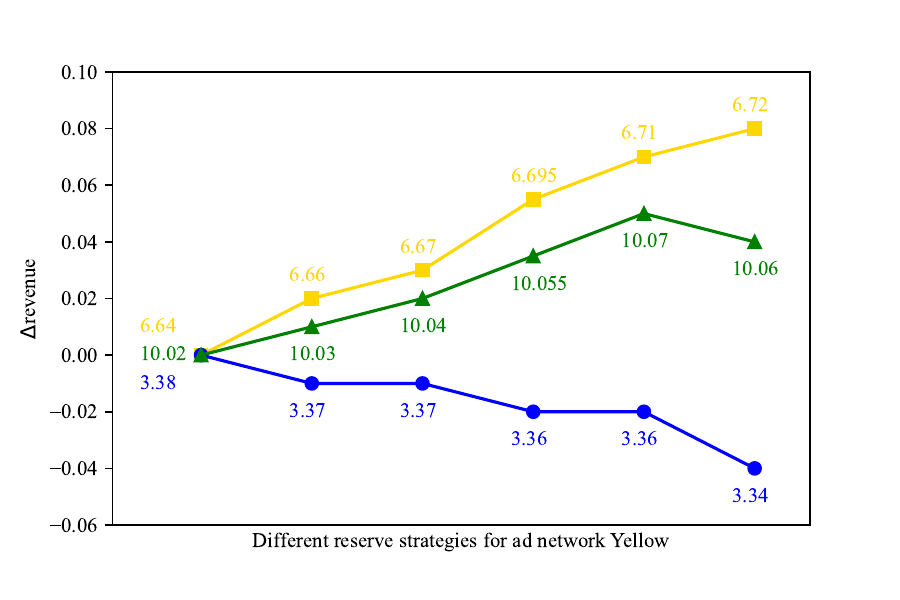}
	\caption{Ad network {\color{blue} Blue} (left) and {\color{Goldenrod} Yellow} (right) apply different reserve pricing strategies to themselves, respectively.} 	\label{fig:ad_networks}
\end{figure*}

The results are given in Figure \ref{fig:ad_networks}.
Ad networks are denoted by colors {\color{blue} Blue} and {\color{Goldenrod} Yellow}, and the platform is their sum {\color{ForestGreen} Green} =  {\color{blue} Blue} + {\color{Goldenrod} Yellow}.
Data points are translated vertically to better compare the net treatment effects, with absolute magnitudes annotated alongside.
In the left plot, Blue applies different reserve pricing strategies to goods of its own.
As its revenue increases, the platform also benefits but at a slower rate since Yellow is hurt.
In particular, Blue's revenue could continue to rise while the platform's keeps flat, which means the revenue it gained comes from pure cannibalization.
The results of the right plot are even worse, as the optimal strategy for Yellow does not align with the optimal one for the platform.

\subsection{Approximate First Price Equilibrium in Second Price Auction Markets}
\label{app:arbitrary_bid}

Apart from the above one-sided exploitability, both bidders and sellers have incentives to deviate from the restriction of multiplicative pacing.
%For bidders, if they were able to do so, bidding spitefully on losing auctions could make the opponents pay more and potentially lower the prices of goods they won.

In the non-monotonicity example, if it were able to bid arbitrarily, Orange could secure good 8 in whole by raising its bids for goods other than 8, such that its opponents, particularly Green, would have to pay more and retreat from the competition for good 8.
Though multiplicative pacing is always enough to give a best response \textit{from hindsight}, it is possible to profitably deviate from it, especially for those at a \textit{disadvantage} at equilibrium: e.g., Orange has the \textit{unique highest} valuation for good 8, but only wins a \textit{fraction} of it; similarly for each pair of variable bidders in the proof of APX-hardness (see Appendix \ref{app:proof_apx_hardness}).
If all bidders are fully strategic, the line between first and second price auction may be blurred:

\begin{proposition} \label{prop:first_best}
	Suppose that each bidder $i$ can bid arbitrary $b_{i, j}$ for every good $j$.
	Then there is an equilibrium where $b_{i, j} = \max_k v_{k, j}, \forall i, j$ and the good is freely shared among those bidders with $v_{i, j} = \max_k v_{k, j}$ at a price of $\max_k v_{k, j}$.
\end{proposition}

If first price auction is used, no auto-bidding is needed and the allocation and payment is exactly the same as described above.
For sellers, any ROI-feasible market outcome (with or without multiplicative pacing) has a revenue upper bounded by its welfare, which is further bounded by the \textit{first-best} one, i.e., the sum of the highest values for each good.
With multiplicative pacing, equilibrium revenue of second price auction is typically less than the first-best welfare by a significant margin,\footnote{See the work by Conitzer et al. \shortcite{conitzer2021pacing}. We do not repeat similar experiments in this paper.} while first price auction naturally achieves it.
%Interestingly, we show that, in second price auction markets where bidders are allowed to bid arbitrarily, there exists an equilibrium attaining the first-best revenue, and recent works (Deng et al. 2021; Balseiro et al. 2021b) on ways to increase revenue could be interpreted as approximating the first-best equilibrium while retaining the single-round IC property.
In second price auction markets, the first-best revenue is achieved at the equilibrium described in Proposition \ref{prop:first_best}, therefore sellers have the incentive to (approximately) do so.
We will illustrate below that how this can be achieved without breaking single-round IC using additive boosts and reserve prices \cite{deng2021towards,balseiro2021robust}.

As mentioned above, with ROI-constraints in place, welfare serves as an upper bound for revenue.
As goods are sold to bidders with lower valuations, the upper bound decreases.
If all bidders' ROI-constraints are binding, the upper bound directly translates to revenue.
Intuitively, additive boosts \cite{deng2021towards} could prevent goods from being allocated to low-value bidders by enlarging the differences of (boosted) valuations.
Bidders with lower valuations would have to bid even higher to win (they will always attempt to do so until binding their ROI-constraints), which could generally elevate the bid landscape and the clearing prices.

Additive boost alone is not enough since it could lower the prices of goods sold to bidders with the highest values.
Deng et al. \shortcite{deng2021towards} did observe a decrease of revenue when boosts are excessively large (while the welfare still increases).
Balseiro et al. \shortcite{balseiro2021robust} fix this with reserve prices, which not only increases the upper bound, but also directly provides a lower bound to avoid cheap sales.

Deng et al. \shortcite{deng2021towards} and Balseiro et al. \shortcite{balseiro2021robust} try to maintain the single-round IC property by imposing that the boosts and reserve prices should be generated from a highly accurate signal on the true (tROI-discounted) valuations  \textit{independently} of behaviors of both auto-bidders and advertisers.
However, if the seller has access to such signals, there is almost no private information left to be elicited from advertisers.
In this case, IC of even the advertiser-game matters little, let alone individual auctions.
In essence, reserve prices behave just like a proxy bidder for the seller.\footnote{This is similar to the \textit{credibility} issue of IC auctions \cite{akbarpour2020credible}. The difference lies in that it may be \textit{implicitly} implemented by widely accepted instruments like additive boosts and reserve prices.}
With more accurate signals, it produces an outcome more similar to the one in Proposition \ref{prop:first_best}.
Boosts effectively does a similar job (though in a less direct manner) such that each bidder faces a more competitive environment in those auctions that it values less than some opponent.

% lose revenue upper bound whenever rank changes
% from winnner's perspective, the opponent bids high
% boost should be independently generated; if dependent, it seems like deviation (descriptive vs prescriptive)
% if the signal is accurate enough, why bother IC since almost no information is private

\section{Utility Instability for Advertisers}
\label{sec:instability}
In this section, we investigate the utility instability problem through an extensive set of experiments.
For equilibrium multiplicity, we generate 5 synthetic instances for each combination of type, distribution and market sizes.
We have three types named \textit{complete, sampled} and \textit{correlated}, and either of the first two comes with a \textit{lognormal} or \textit{uniform} distribution (and will be called, e.g., complete lognormal) while correlated ones are parameterized by $\sigma$ that characterizes the intensity of competition.
See Appendix \ref{app:market_instances} on how they are generated.
Instances of \textit{small} sizes have $n = 10, m = 14$ and are solved using MIBLP, and \textit{moderate} sizes have $n = 50, m \in \{100, 200, 300\}$ and solved using the iterative method.
For valuation sensitivity, instances are constructed by modifying the same base realistic instance, which consists of 85596 bidders, 3118648 goods and 14967663 non-zero $v_{i, j}$.
%Among this large number of bidders, about half could win at least one good.

\begin{figure*}
	\centering
	\includegraphics[width=0.85\textwidth]{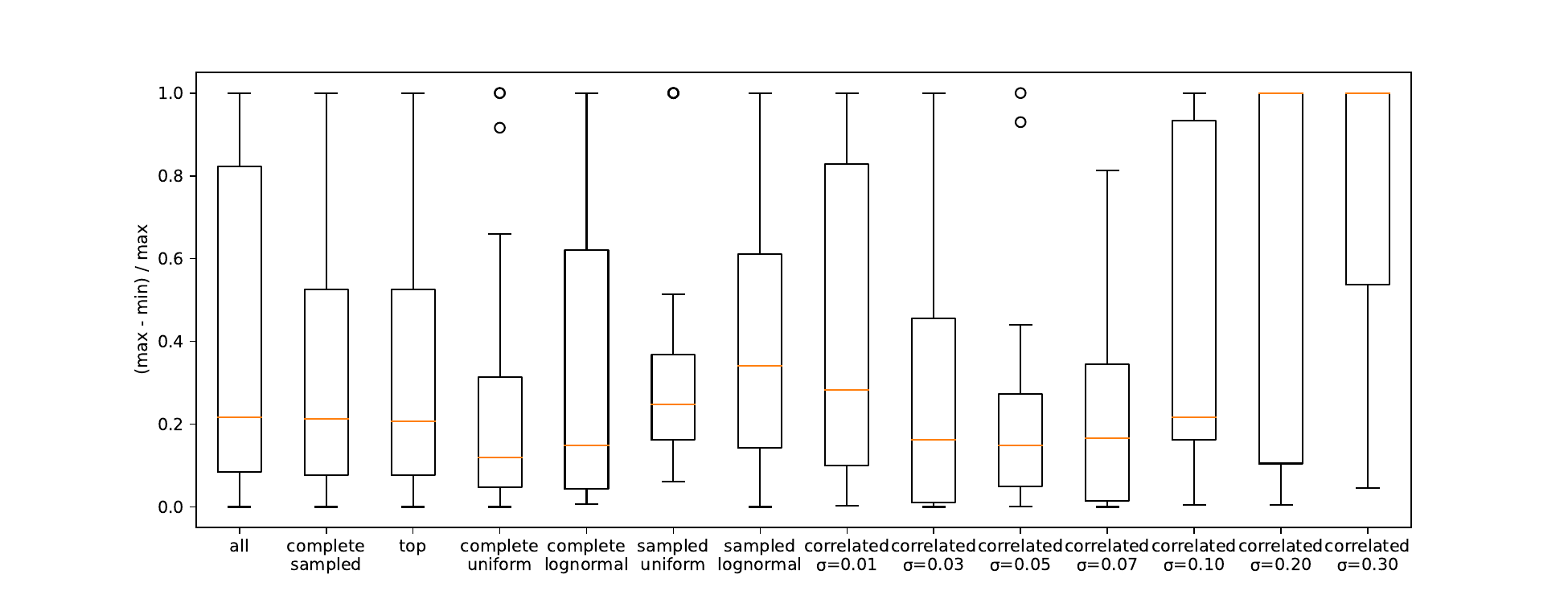}
	\caption{Gap distribution of small instances.} \label{fig:instability_small}
\end{figure*}

\subsection{Equilibrium Multiplicity}
\label{app:multiplicity}

In this section, the \textit{gap} of bidder $i$ is measured as ``the difference between the maximum and minimum values received by bidder $i$ in any equilibrium'' divided by ``the maximum value received by bidder $i$ in any equilibrium.''
Bidders winning nothing in every equilibrium will not be considered.

Experiment results of small instances are given in Figure \ref{fig:instability_small}.
Over all small instances, more than half of bidders have a gap of more than 20\%, and more than a quarter have a gap more than 80\% (see the first box plot).
Even though many extreme cases are contributed by correlated instances, there is still more than a quarter of bidders have a gap of more than 50\% for complete and sampled instances (the second box plot).
We also check the gap distribution for the top 3 bidders of each instance (the third box plot), measured in their acquired valuations if first price auction is used (to make the ranking unique).
They do perform better than the rest, but not much.

The gap distributions differ across classes of instances, with lognormal generally worse than uniform, sampled worse than complete.
Correlated instances exhibit a U-shape w.r.t. $\sigma$: the gap is generally larger for both highly intense and sparse competition, and relatively small for middle ones.
Nonetheless, at least a quarter of bidders always suffer a gap of about 30\%, and there are almost always some unlucky bidders that win something in an equilibrium, but lose completely in another.

\begin{figure*}
	\centering
	\includegraphics[width=\textwidth]{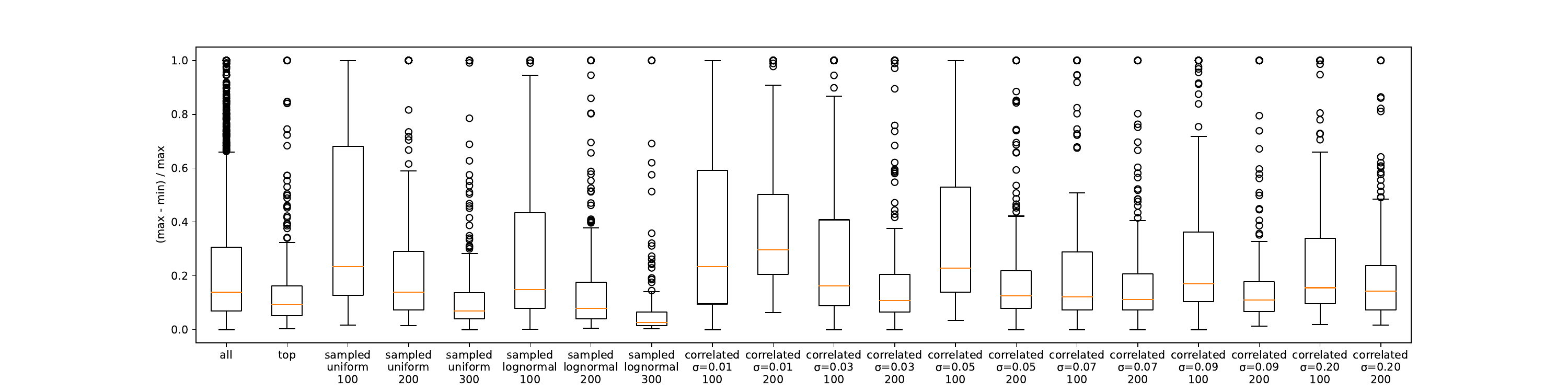}
	\caption{Gap distribution of moderate instances.} \label{fig:instability_moderate}
\end{figure*}

The situation improves for moderate instances, as shown in Figure \ref{fig:instability_moderate}.
We can actually observe a clear trend where the gap distribution contracts as the number of goods increases.
However, large gaps still happen a lot.
Note that, even for instances of the same type and distribution, different sizes may also give quite distinct market structures.
For an evidence, fixing $n = 50$, if $m = 100$, no more than 40 bidders on average could win something in at least one equilibrium.
This number increases to around 45 for 200 goods, and almost 50 for 800 goods (not in the figure).

\subsection{Valuation Sensitivity}
\label{app:sensitivity}

We examine sensitivity of two types, one from an individual's point of view, and the other from the market's.
Figure \ref{fig:sensitivity} consists of two parts: the central box plot and the rest four, each corresponding to one type of sensitivity.

\begin{figure}[h]
	\centering
	\includegraphics[width=\columnwidth]{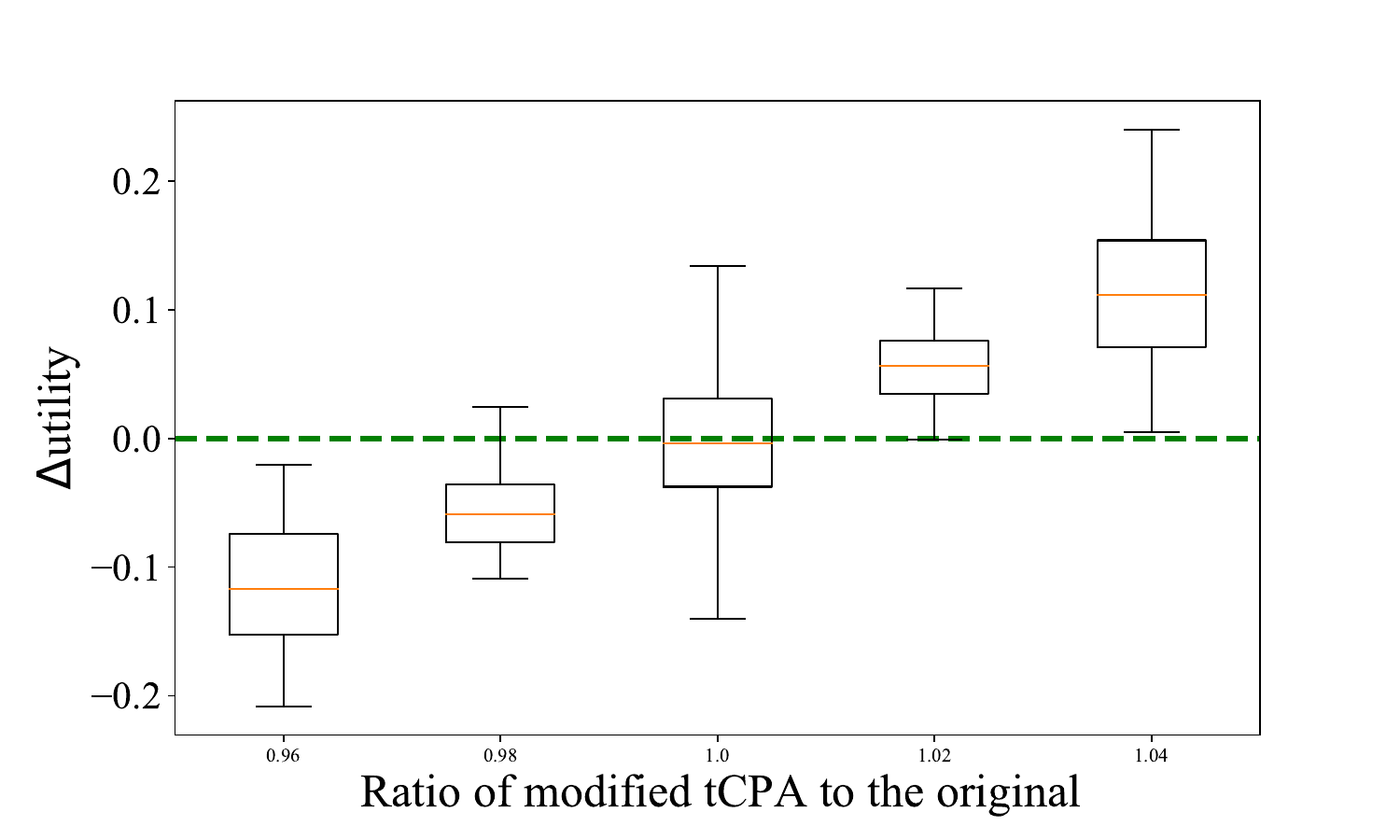}
	\caption{Change of utilities after adding noises. Outliers beyond whiskers are excluded.}
	\label{fig:sensitivity}	
\end{figure}

\begin{figure}[h]
	\centering
	\includegraphics[width=\columnwidth]{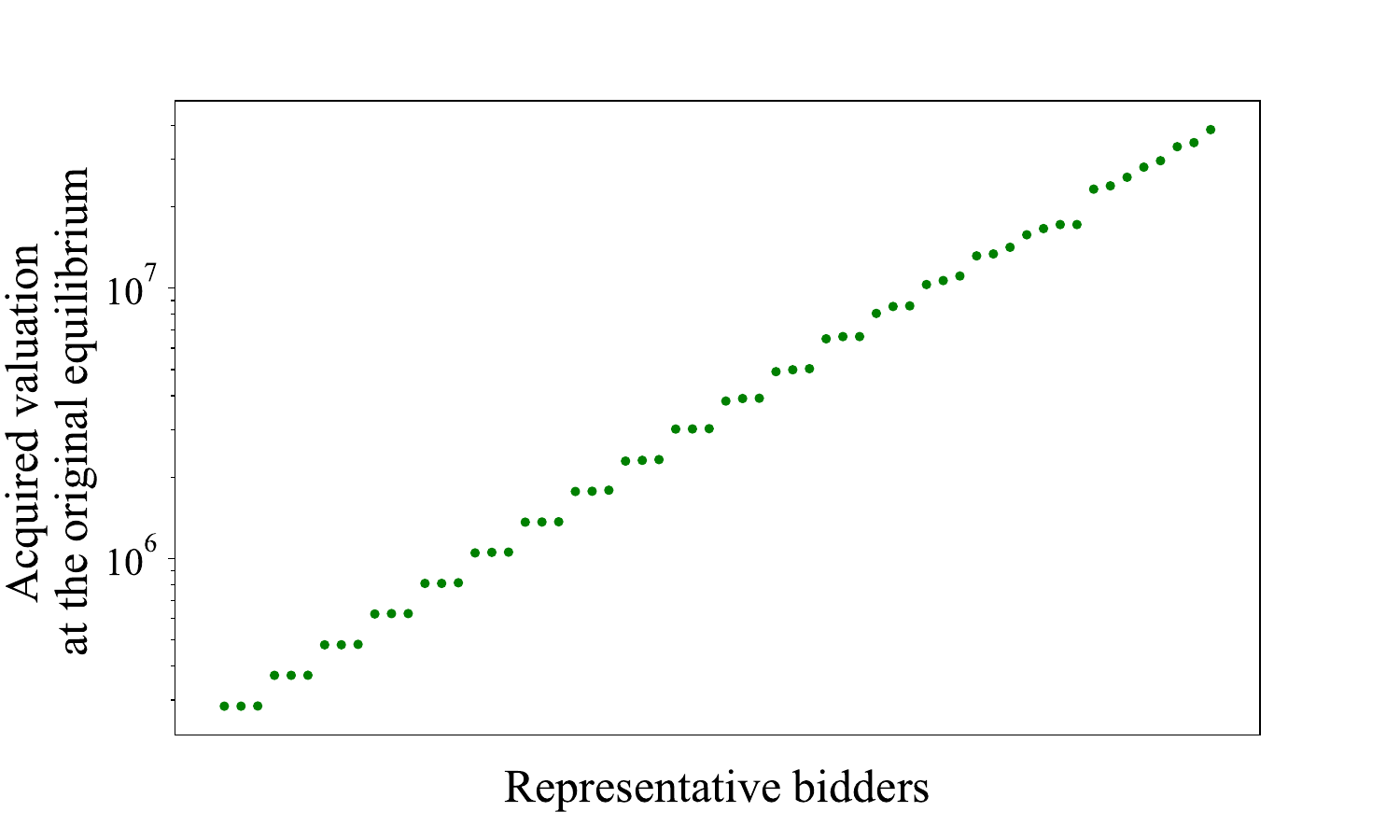}
	\caption{Selected bidders' acquired valuations in the original equilibrium, which are made distributed evenly in the log scale.}
	\label{fig:representative_bidders}	
\end{figure}

For the four non-central box plots, each time we choose a bidder whose tCPA is modified to $x \in \{0.96, 0.98, 1.02, 1.04\}$ times of the original\footnote{In practice, these perturbations can happen when advertisers directly change their tROIs, or could be unintentionally brought by noises of machine learning models, which could bias towards a single advertiser (non-central four box plots) or spread across the whole market (the central one).}
and observe the difference between the equilibrium before and after the change.
To ensure that the selected bidders are representative, we group them into 20 buckets based on their acquired valuations in the old equilibrium, and select three from each bucket (see Figure \ref{fig:representative_bidders}).
First, we find that non-monotonicity still exists, as shown by the whisker above zero for $x = 0.98$.
Second, it holds for a significant portion of bidders that their utilities are quite sensitive to tCPAs. A quarter of bidders can gain or lose more than 15\% of its value by making only a change of 4\% of its tCPA.
In the extreme case,  a 4\% change can bring 45\% gain or 46\% loss.

For the central box plot, we randomly generate 10 modified instances, each of which is constructed by flipping a fair coin for every bidder and adjusting its tCPA to either 0.99 or 1.01 times the original accordingly.
We only focus on the top 4000 bidders measured in their acquired valuations at the original equilibrium.
The gap distribution seems more stable than changing tCPA individually, but it features more extreme cases: about 3\% of bidder-modification pairs have a gap of more than 20\%, and more than 0.5\% have a gap of more than 50\%.

It is worth noting that markets with budget-constrained value-maximizers (in both first and second price auction markets) are completely insensitive to this kind of perturbation that scales a bidder's valuation by a uniform factor.
It has indeed been observed in practice that budget-constrained campaigns outperform ROI-constrained ones in stability.

\subsection{Discussions}
\label{app:discussion_instability}

\textit{When will the utility become more robust against equilibrium selection?}
As implied by our hardness results, it is generally hard to compute how unstable the utility of each bidder is.
Based on the U-shape results of small correlated instances, we hypothesize that:
(1)
When the competition is intense, the boundary between win and loss becomes narrow; when competition is sparse, the winner gains a large value surplus if the opponents' multipliers are low.
In both cases, it is more profitable to commit a large multiplier to claim the high ground in the competition and thwart the opponents.
%Hence in this situation the utility is more unstable.
(2)
With a moderate variance, the winner does not possess a large surplus even if multipliers of opponents stay at 1, and no one can suppress others easily.
In this case, valuations are distributed more evenly across bidders and the utility becomes more stable.

An even value distribution across bidders, or a larger number of competitive bidders does seem to bring more stability with respect to equilibrium multiplicity.
This is further supported by the results of moderate instances, where utility becomes more stable as the number of goods increases but the competition for each good remains intact.
Intuitively, with more competitive bidders, they clamp each other together to form a more stable structure.
However, in this case the market may be more sensitive to the inputs, as shown by the realistic instances.

\textit{Why does duplication work in practice?}
Utility instability has been widely observed in practice.
Advertisers have found a simple countermeasure called \textit{duplication}, which is to create many campaigns with approximately the same configuration.
Platforms will typically prevent auto-bidders of the same advertiser from competing with each other.
Then with more campaigns, the instability of valuation prediction and equilibrium selection could be better neutralized, and the aggregate performance could be smoother.
Our method to find multiple equilibria with the iterative method (Appendix \ref{app:iterative}) is actually motivated by the real-world observation that, at the beginning of a campaign's life, the prediction is more noisy, and duplication will effectively drive up the bid.

Duplication is annoying for platforms since it imposes an extra burden on the system but adds almost nothing informative and valuable.
Some platforms have restricted the maximum number of campaigns that can be run simultaneously by a single advertiser, and algorithms have been implemented to cluster and prune similar campaigns.
To alleviate instability, platforms could actively  boost or suppress bidders (as  in our experiments) based on their historic performance.
Note that, when we apply the Bernoulli $\pm 0.01$ noise to valuations, the total revenue remains stable, with 
a gap between the worst and best only $0.15\%$.
Even when we increase the noise magnitude to 0.02, the gap is just $0.22\%$.
Hence such strategies will not sacrifice the revenue of the platform.

\begin{figure*}[h]
	\centering
	\includegraphics[width=0.99\columnwidth]{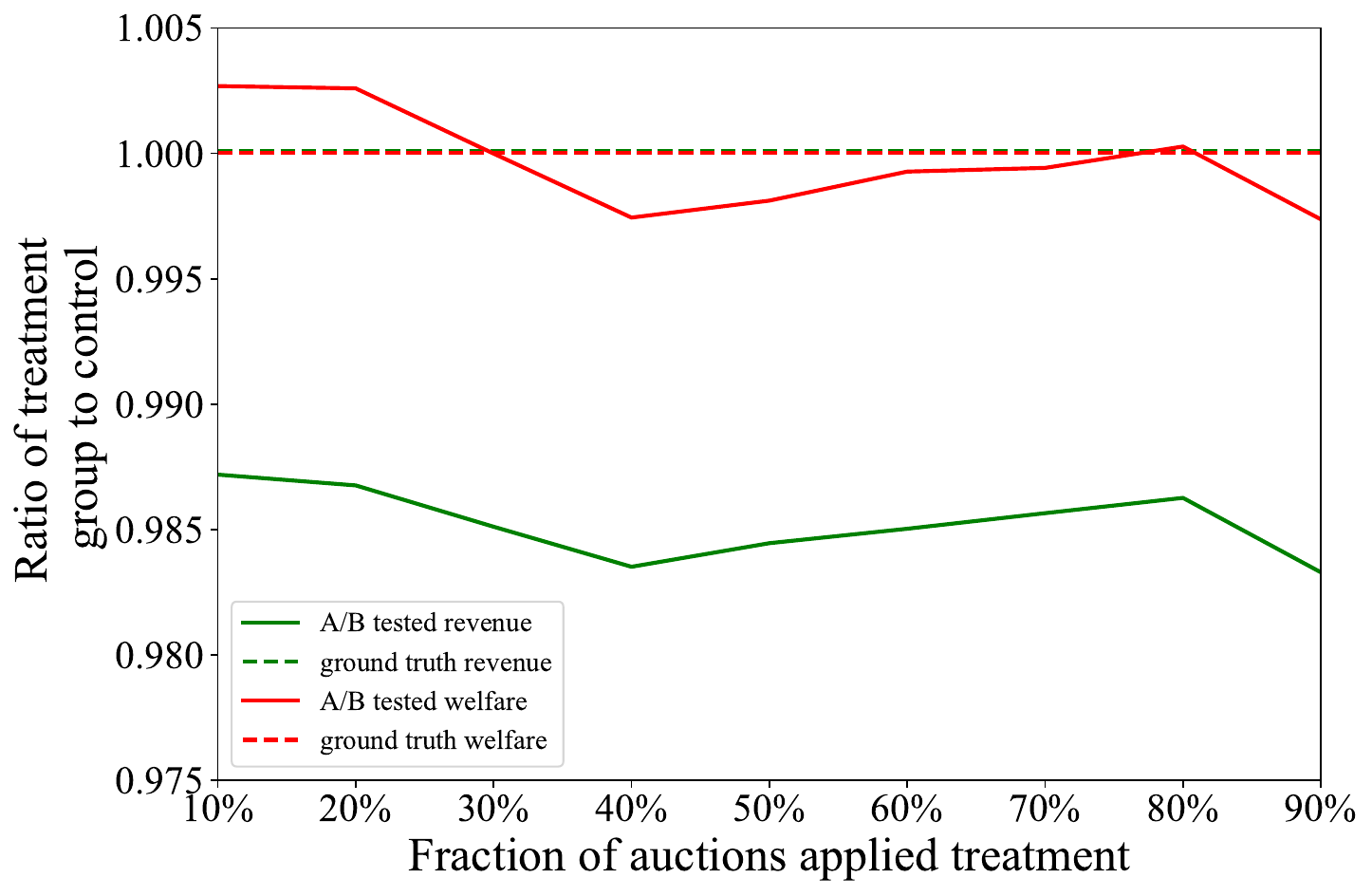}
	\includegraphics[width=0.99\columnwidth]{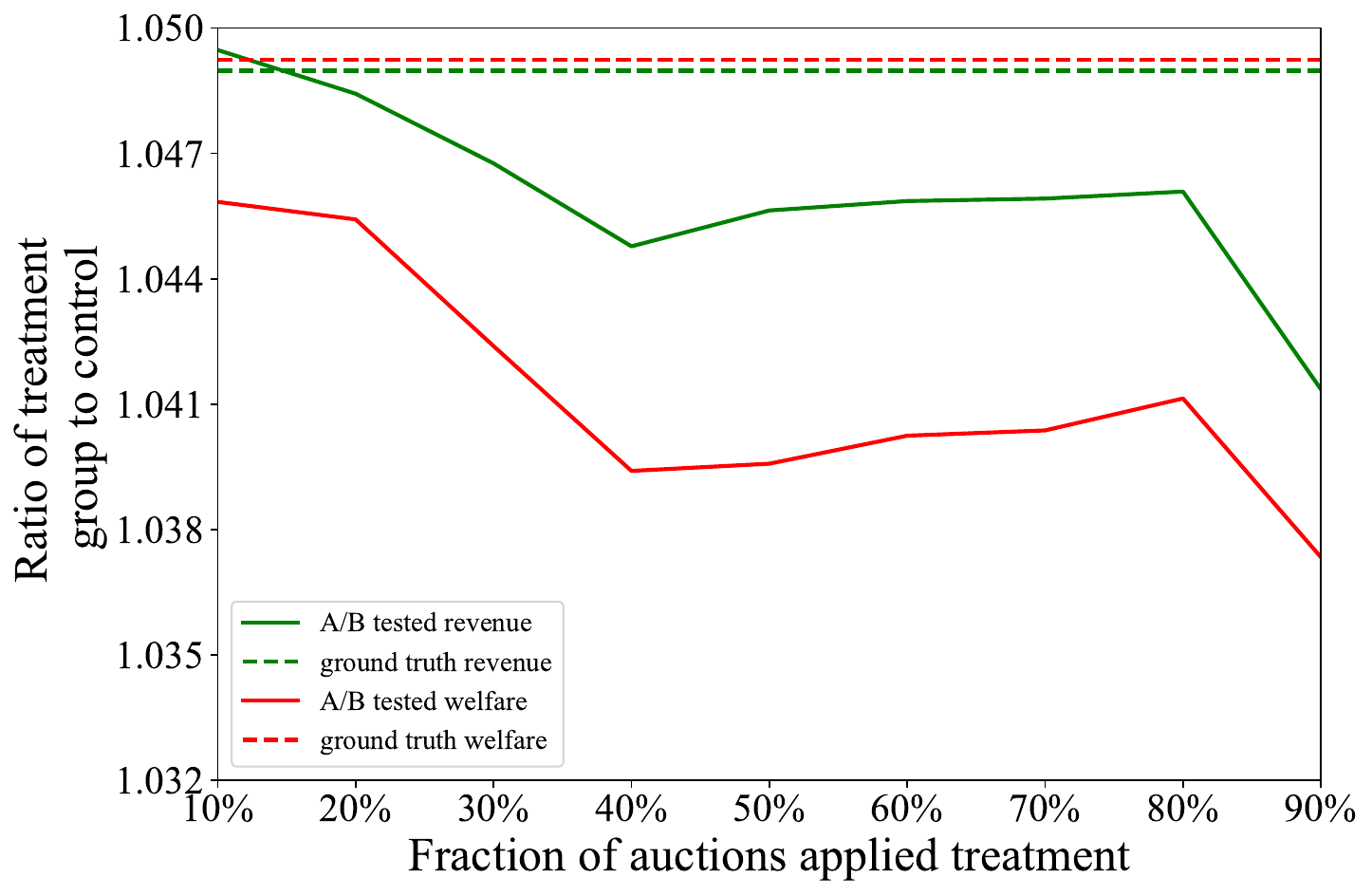}
	\caption{Results of user-side experiments.} 	\label{fig:user_ab}
\end{figure*}

\section{A/B Testing for Platforms} \label{app:ab_testing}
In this section, we investigate the interference in A/B testing of auction markets.
User-side experiments are performed by applying the treatment variant to $r$ percents of auctions ($x$-axis in Figure \ref{fig:user_ab}; for each $r$, results are aggregated over several random partitions to make sure that the biases are not from non-homogeneous user-splitting) uniformly at random and the control variant to the rest.
We will measure the ground-truth treatment effect (dashed lines in Figure \ref{fig:user_ab}) by computing the ratio of  ``metrics when applying the treatment variant to all the auctions'' over ``metrics when applying the control variant to all the auctions.''
The A/B tested treatment effect (real lines in Figure \ref{fig:user_ab}) is the ratio of metrics between experiment groups.
We use realistic bidding data and include the additive boost mechanism.
In the first experiment (left of Figure \ref{fig:user_ab}), the treatment group is assigned a boost of magnitudes $c_{i, j} = 0.05v_{i, j}$ and the control group has no boost.
%\footnote{\cite{balseiro2021robust} is meant to increase revenue and welfare, though a coefficient of $0.05$ seems not enough to achieve the goal.}
In the second, we use the boost used in practice for the control group, and half its value for the treatment.
The real-world boost is used by the platform as a measure of long-term utility to the platform. As we can see from Figure \ref{fig:user_ab} that, when it is halved, the \textit{immediate} revenue is indeed increased by sacrificing long-term utilities, and thus its weight should be carefully configured to optimize the overall financial consequences.

%For ad-side experiments, we deviate a little from the deterministic market model to a stochastic one that better resembles situations where ad-side experiments are used in practice, but the result can still be explained largely within our framework.
For ad-side experiments, to better resemble the situation where ad-side A/B testing is used in the real world, here we deviate a little from our model (but the result can still be explained largely within the framework).
In practice, it is important for a bidding algorithm to handle the uncertainty of the environment.
Therefore we incorporate some stochastic features into the market as follows.
Each market instance consists of 200 bidders and $400 \times 1000$ goods.
Auctions do not happen simultaneously, but arrive one-by-one.
Before each auction $j$,  each bidder $i$ has an unbiased prior estimation $v_{i, j}$ for the value of the good.
In the simulation, $v_{i, j}$ is drawn i.i.d. from the distribution $\min\bigparen{\abs{X}, 1.0}), X \sim \text{Normal}(0, 0.1)$.
The value received by bidder $i$ after winning is \textit{not} $v_{i, j}$, but randomly drawn from $\text{Bernoulli}(v_{i, j})$ (simulating conversions of discrete types in practice).
For each bidder, its bid $b_{i, j} = \alpha_i v_{i, j}$ is still given by the multiplicative pacing strategy, where $\alpha_i$ is controlled by a PID controller (see the work by Zhang et al. \shortcite{zhang2016feedback}) that adjusts the multiplier after every 1000 auctions to optimize the acquired \textit{realized} value subject to that the ROI-constraint should be satisfied approximately.
PID controllers have 4 parameters and we will run A/B tests for different pairs of parameter configurations (100 bidders for each group, and bidders in the same group share the same parameter configuration).

Ad-side experiment groups are evaluated on two classes of metrics: efficiency and stability, both are widely used by practitioners.
Though the market becomes stochastic and dynamic, it is easy to see that, in the long run, multipliers should still converge to the auto-bidding equilibrium for the market with the unbiased estimations $\{v_{i, j}\}$.
Efficiency metrics (the revenue and welfare of the platform) are meant to measure how well the resulting equilibrium performs economically.
On the other hand, since $\{v_{i, j}\}$ is only an expectation and goods arrive in an online fashion, ROI-constraints may not always be perfectly satisfied from hindsight.
So we use the success rate (the proportion of \textit{revenue} generated from bidders ending with a ROAS between 0.975 to 1.025) and failure rate (the proportion of revenue from bidders with ROAS above 1.025) to quantify how well the PID controllers handle the uncertainty for advertisers.
Each pair of parameter configurations is evaluated through A/B testing over 100 market instances, and statistical test is performed with respect to each metrics to determine whether the treatment variant performs significantly better or worse than the control.
The ground-truth is measured by applying the configuration to all bidders and aggregate each metrics over 100 market instances.

%Although the setting is not static (see appendix \ref{app:ab_testing}), the bias can be explained by our concept of equilibrium and (kind of) utility instability.
%In a naive test, the treatment group competes the control group as a whole.
%And a better test efficiency does not correspond to a positive treatment effect, but means that, in the resulted equilibrium, the treatment group performs better than the control.
%For bias of stability, one hypothesis is that, if many more goods are won by one group, the ex-post value would converge better to the expectation, and the environment becomes less uncertain.

\begin{table*}
	\centering
	\begin{tabular}{c c  c  c  c  c}
		\toprule
		naive vs. ground truth &
		setup & revenue & welfare & success rate & failure rate \\
		\midrule
		\multirow{3}{3.6cm}{overestimate efficiency wrong stability}
		& ground truth
		& +1.21\%*
		& +1.10\%*
		& -3.96\%*
		& +11.55\%*
		\\
		& naive
		& +13.18\%*
		& +13.37\%*
		& +2.76\%*
		& -9.31\%*
		\\
		& boosted
		& +0.86\%*
		& +0.91\%*
		& -1.91\%*
		& +4.60\%*
		\\ 
		\multirow{3}{3.3cm}{wrong efficiency wrong stability}
		&
		ground truth
		& -0.67\%*
		& -0.76\%*
		& +2.96\%*
		& -2.10\%
		\\
		& naive
		& +3.12\%*
		& +3.06\%*
		& -0.67\%
		& +3.33\%
		\\
		& boosted
		& -0.64\%*
		& -0.73\%*
		& +0.89\%*
		& +0.86\%
		\\ 
		\multirow{3}{3.3cm}{wrong efficiency wrong stability}
		& ground truth
		& -0.29\%*
		& -0.31\%*
		& -3.56\%*
		& +3.56\%*
		\\
		& naive
		& +13.49\%*
		& +13.36\%*
		& +5.28\%*
		& -1.61\%
		\\
		& boosted
		& -0.24\%*
		& -0.30\%*
		& -4.84\%*
		& +8.38\%*
		\\ 
		\bottomrule
	\end{tabular}
	\caption{Results of three groups of ad-side experiments. Each entry is calculated as the ratio of ``metrics difference of the treatment variant to the control'' over ``the metrics of the control variant''.
		Asterisks indicate statistical significance. Note that a strategy with \textit{lower} failure rate is considered better.} \label{tab:ad_ab}
\end{table*}%

\subsection{Simulation Results}
\label{subsec:ab_testing_results}

\paragraph{User-side experiments.}
Figure \ref{fig:user_ab} depicts results of user-side experiments for two pairs of treatment/control variants.
In the left, the ground-truth treatment effect is nearly zero, but A/B testing shows that we may lose more than 1\% revenue.
Welfare seems more robust in this case, but it is also inconsistent and has a gap of more than 0.5\% between the most optimistic and pessimistic tests.
In the second experiment,  welfare deviates worse than revenue, with a maximum bias of more than 1\%.
Note that a bias of 0.5\% is large enough to affect practitioners' decisions (e.g., the minimum detectable treatment effect in the platform from which our data are taken is roughly at the same scale), and if multiple objectives are to be weighted, the requirement for accuracy is even stricter.
Based on our observation over some more experiments not displayed here, we find that welfare is generally more robust than revenue.
If naive A/B testing is the only option, welfare seems to be more reliable as an overall evaluation metrics.

\paragraph{Ad-side experiments.}
Examples with biases of different types are given in Table \ref{tab:ad_ab} (naive vs. ground truth).
The bias is irregular as we can observe that there may be significant overestimation, or a worse variant may be measured as better, etc.
In particular, in the third example, naive A/B tests tell us that the treatment variant beats the control variant in all metrics by a large margin, but if we apply it to all bidders, the resulting system would instead perform worse.

\subsection{Explanations and a General Approach to Diagnosis and Design}
\label{subsec:invariant}

During user-side A/B testing, the market is a mixture of treatment and control groups, and the resulting multipliers at equilibrium are not equal or even close to those if either variant were applied counterfactually to all the auctions.
The same argument can also explain the failure of ad-side A/B tests on measuring revenue and welfare.
As for success and failure rate, the hypothesis is that, if a significantly more amount of auctions are won by one group (as is the case here), their received values would converge closer to the expectation and the environment becomes less uncertain and easier to deal with.
In summary, in an A/B test, any experiment group (of either users or ads) faces a different environment from the one in the counterfactual ``B/B'' or ``A/A'' test where both groups were assigned the same variant.
Thus to reduce bias, a general principle  is to make the experiment group behave (for each individual bidder/good or at least the group as a whole) \textit{as if} the other group is also assigned the same variant.

The principle might seem too general and it is not immediately clear how to achieve it, in particular in markets involving complex strategic behaviors.
To make it applicable, we propose an approach using the concept of \textit{invariants}.
We start with identifying a metrics as the invariant, which should remain the same regardless of whether the treatment or control variant is applied.
For naive tests, it serves as a signal that there may be bias if this metrics \textit{varies} significantly across experiment groups.
For example, the ROI-constraint should always be binding when aggregated across the market.
But in Figure \ref{fig:user_ab} the ratio of revenue over welfare for each group deviates significantly from 1, which is a signal of bias worth being taken noticed of by practitioners.
As for designing less-biased setups, we should try to restore the balance of the invariant.
In retrospect, the work by Liu et al. \shortcite{liu2020trustworthy} is an instantiation of this approach for user-side experiments where the invariant is the budget spent by each advertiser, and the work by Ha et al. \shortcite{ha2020counterfactual} is one for  non-auction markets with recommendation positions (occupied by each group as a whole) as the invariant.

For user-side experiments in ROI-constrained markets, ROI simply serves as the invariant and we can assign an auto-bidder for each advertiser \textit{and} each experiment group (instead of the whole market as more commonly done in practice).
This computes the correct outcome in theory if there is a unique (and efficiently computable) equilibrium and the two experiment groups are perfectly homogeneous.
%Practitioners may also compute equilibrium multipliers in a counterfactual way  to determine whether there is a mismatch between results of A/B testing and the ground truth.
For the less-studied ad-side experiments, we instantiate the approach with a new \textit{boosted design} using the \textit{number of auctions won by each group on average} as the invariant.
To make it balanced, right after each episode (here 1000 auctions), a uniform additive boost $c$ for treatment group will be calculated in a counterfactual way such that, if it were applied to the treatment group in the last episode (i.e., if good $j$ were allocated to the bidder with highest $\alpha_i v_{i, j} + c_i$, where $c_i = c$ if $i$ is in treatment group, $c_i = 0$ otherwise), goods would be  split into two groups perfectly in proportion ($r$ percent of ads should win $r$ percent of auctions). 
%(note that in an A/A or B/B test, goods should be distributed this way in expectation). %If there is an interval of values satisfying the above property, take the middle point.
The boosts will then be applied to the treatment group throughout the next episode.

Table \ref{tab:ad_ab} (with the setup ``boosted'') shows that it works really well for these synthetic instances: it always gives the correct qualitative result (statistically better, worse, or insignificant) and it is also more accurate quantitatively, especially on revenue and welfare.
As a comparison, we try to adapt the counterfactual design by Hu et al. \shortcite{ha2020counterfactual} where treatment and control are applied to all bidders respectively to get two rankings, and then merge them to decide the winner.
Ha et al. \shortcite{ha2020counterfactual} (not designed for auctions) did not specify the payment rule, and we tried several alternatives but no one worked: in the contrary, the number of auctions won by each group is more unbalanced and the bias is severely exaggerated.
%A good design may not directly translate to other settings.
%This may serve as a further support on the use of invariants since a good design may not translate to other settings.

%One of the reasons may be that the ex-post value is generated from the \textit{merged} ranking, and the consequent responses of either treatment or control are still different from those if a single one is applied.

\section{First Price Meta-game Equilibrium with Frugal ROI-constrained Value-maximizers}
\label{app:frugal}

Let $(v_1, p_1)$ and $(v_2, p_2)$ be two outcomes where $v_k$ is the total value received by some bidder $i$ and $p_k$ the payment.
If $i$ is a \textit{frugal} ROI-constrained value-maximizing bidder, it strictly prefers $(v_1, p_1)$ to $(v_2, p_2)$ if and only if either (1) $v_1 > v_2$ or (2) $v_1 = v_2$ and $p_1 < p_2$.

In the main text we assume a tROI of zero.
Here we will consider advertisers' behavior in the meta-game, and we will use  $v_{i, j}$ to denote the valuation of advertiser $i$ to good $j$ \textit{discounted} by $i$'s true tROI rather than the reported.
For simplicity we also assume $\argmax_{i} v_{i, j}$ is unique for every good $j$.
\begin{definition}[First Price Meta-game Equilibrium]
	The equilibrium $(x, p)$ consists of allocation $x \in \{0, 1\}^{n \times m}$ and payment $p \in [0, +\infty)^n$ such that
	\begin{itemize}
		\item goods go to the bidder with the highest true tROI-discounted valuation: $x_{i, j} = 1$ if and only if $i = \argmax_k v_{k, j}$, $\forall j$;
		\item bidder pays its total acquired valuation discounted by the minimum tROI required to win the set of goods currently allocated:
		\begin{displaymath}
			p_i = \alpha_i \sum_j v_{i, j} x_{i, j}
			\text{ where }
			\alpha_i = \max_{k, j: x_{i, j} = 1, k \neq i} \frac{ v_{k, j}}{v_{i, j}}.
		\end{displaymath}
	\end{itemize}
\end{definition}

At equilibrium, each advertiser effectively bids $b_{i, j} = \alpha_i v_{i, j}$ in each individual first price auction and ties are broken in favor of the one with the highest (true tROI-discounted) value.
If $i$ bids with a multiplier less than $\alpha_i$, some other bidder could bid truthfully and win an item that is currently allocated to $i$.

The equilibrium is unique and achieves  the first-best welfare (discounted by true tROIs).
It is not clear how to compare the revenue with second price auction, wherein the behaviors of both advertisers and auto-bidders are highly unpredictable.

\section{Fairness and Competitive Equilibrium}
\label{app:fairness}

Budget-pacing equilibrium in first price auction markets with budget-constrained \textit{value}-maximizers is exactly equivalent to the classic competitive equilibrium in linear Fisher markets (we do not have a reference but the correspondence is straightforward from the definition). Conitzer et al. \shortcite{conitzer2021multiplicative,conitzer2021pacing} relate their budget-pacing equilibrium for both first and second price auctions to a modified version of competitive equilibrium for budget-constrained quasi-linear utility-maximizers. First price equilibrium is mathematically closer to the classic competitive equilibrium as the bang-per-bucks of all items allocated to a bidder are equalized, while second price equilibrium requires the buyer to be \textit{supply-aware}. Nonetheless, both outcomes are Pareto-optimal and there seems no consensus on whether one is fairer than the other. Traditionally there is no ROI-constrained version of competitive equilibrium, so we do not discuss this in this paper. But it is easy to see that similar results (e.g., Pareto-optimality) exist with proper definition.

\section{Case Study: Google AdSense's Partial Shift of Auction Format}
\label{app:google_shift}

We have given much mathematical evidence in the main text to establish the dominance of first price auction over second price \textit{within our auto-bidding market model}.
But once an advertiser could submit tROIs (or budgets) for even just two different sub-markets of the platform, it opens opportunities for strategic manipulations and many aforementioned properties technically break.
This, however, does not mean that second price auction could immediately regain the upper hand as in the single-item regime.
Though it cannot be easily stated and verified in a fully rigorous manner, it is reasonable to hypothesize that: with advertisers having access to a wider range of heterogeneous sub-markets, second price auction could reduce strategic behaviors more effectively and bring better market outcomes; conversely, if optional sub-markets are coarse-grained and advertisers do not have enough knowledge to differentiate them, the market could still be well captured by our model and first price auction mostly retains its dominant position.

% for content, video and game; the content distribution is purely centralized; machine learning model knows you better than yourself
% advertisers have restricted information (due to privacy) and knowledge (due to algorithms) to bid among cohorts
% for search words, however, users come with clear intents
% advertisers have access to the information (due to brand safety or legacy) and knowledge (ultimate value)

As for Google's partial shift of auction format,
the most distinctive feature that sets Search and Shopping apart from Content, Video and Games is the way in which items (both sponsored and non-sponsored) are distributed.
For the former, users come with a clear \textit{intent} expressed through the \textit{keywords} they input.
Even though there is still personalization, keywords largely determine the items to be recommended,  and each keyword characterizes a sub-market of which advertisers have a fairly good understanding.
In contrast, for Content, Video and Games, there is typically no such context (available to advertisers) and users' general interests are highly personal.
%A platform can learn a lot from users' interactions with items, but disclosure of such information is extremely selective for both privacy and economic concerns.
The platform can learn a lot of information from users' interactions with the items, while advertisers cannot manipulate much to optimize every penny they spend in this more centralized market.
In this case, first price auction should be more desirable.

\end{document}